\documentclass[draftclsnofoot, onecolumn, 12pt]{IEEEtran}

\usepackage{floatflt}

\usepackage{color}
\usepackage{cite}

\usepackage{graphicx}
\usepackage{amsmath}
\usepackage{amsthm}
\usepackage{amssymb}

\usepackage{verbatim}

\usepackage{psfrag,array}

\usepackage{subfigure}
\usepackage{multirow}
\usepackage{hhline}

\usepackage{stfloats}

\usepackage{amsmath}
\usepackage{epstopdf}
\usepackage{algorithmic}

\newcommand{\p}[1]{\mathop{\mbox{\it p} } }

\renewcommand{\vec}[1]{\ensuremath{\boldsymbol{#1}}}
\newcommand{\be}{\begin{equation}}
\newcommand{\ee}{\end{equation}}
\newcommand{\ba}{\begin{array}}
\newcommand{\ea}{\end{array}}
\newcommand{\bea}{\begin{eqnarray}}
\newcommand{\eea}{\end{eqnarray}}
\newcommand{\bean}{\begin{eqnarray*}}
\newcommand{\eean}{\end{eqnarray*}}

\newcommand{\rmh}{^{\dag}}
\newcommand{\rmt}{^{\rm T}}

\newcommand{\scell}[2][c]{\begin{tabular}[#1]{@{}c@{}}#2\end{tabular}}

\definecolor{white}{rgb}{1,1,1}

\newtheorem{theorem}{Theorem}
\newtheorem{lemma}{Lemma}
\newtheorem{example}{Example}

\newtheorem{property}{Property}

\newtheorem{corollary}{Corollary}

\begin{document}

\title{Optimal Channel Shortener Design for Reduced-State Soft-Output Viterbi Equalizer in Single-Carrier Systems}

\author
{
\begin{tabular}{c}
Sha Hu$^{\dagger}$, Harald Kr\"{o}ll$^{\ddagger}$, Qiuting Huang$^{\ddagger}$, and Fredrik Rusek$^{\dagger}$ 
\end{tabular}
\thanks{$^{\dagger}$Department of Electrical and Information Technology, Lund University, Lund, Sweden (email: firstname.lastname@eit.lth.se).

$^{\ddagger}$Integrated Systems Laboratory, ETH Z\"{u}rich, Switzerland (email: \{kroell, huang\}@iis.ee.ethz.ch).}
}

\maketitle

\begin{abstract}
We consider optimal channel shortener design for reduced-state soft-output Viterbi equalizer (RS-SOVE) in single-carrier (SC) systems. To use RS-SOVE, three receiver filters need to be designed: a prefilter, a target response and a feedback filter. The collection of these three filters are commonly referred to as the \lq\lq{}channel shortener\rq\rq{}. Conventionally, the channel shortener is designed to transform an intersymbol interference (ISI) channel into an equivalent minimum-phase equivalent form. In this paper, we design the channel shortener to maximize a mutual information lower bound (MILB) based on a mismatched detection model. By taking the decision-feedback quality in the RS-SOVE into consideration, the prefilter and feedback filter are found in closed forms, while the target response is optimized via a gradient-ascending approach with the gradient explicitly derived. The information theoretical properties of the proposed channel shortener are analyzed. Moreover, we show through numerical results that, the proposed channel shortener design achieves superior detection performance compared to previous channel shortener designs at medium and high code-rates. 
\end{abstract}


\section{Introduction}

Communication systems based on single carrier (SC) modulation are currently used in 2G networks~\cite{H16} which have the largest number of subscribers worldwide. Besides personal mobile communication they play a key role in the latest LTE-Advanced Release 13 \cite{3gpp36201}, where on the path to 5G Internet of Things (IoT) networks, the standard EC-GSM-IoT was released together with SC waveforms for high power efficiency requirements \cite{IoT}. Moreover, SC modulation is also used in satellite communications and high-speed serial links\cite{GF12}. The advantages of a low peak-to-average-power ratio (PAPR), low device complexity, straightforward synchronization, and the absence of cyclic-prefix (CP) overhead favor its use in many low data rate scenarios over multi-carrier (MC) systems\cite{CD15, FASE02}. However, SC systems suffer from intersymbol interference (ISI) caused by delay dispersion along the multi-path propagation from the transmitter to the receiver.

In order to combat intersymbol interference (ISI) caused by delay dispersion in propagation channels in SC systems, Forney proposed the Viterbi algorithm (VA) \cite{F72} that implements maximum log-likelihood sequence estimation (MLSE). With error correcting codes such as turbo codes\cite{turbo}, low-density parity-check (LDPC) codes\cite{ldpc}, and polar codes\cite{polar}, it is well-known that soft-decisions output from the equalizer, i.e., the reliability information, are superior to hard-decisions. In \cite{sova}, Hagenauer and Hoeher modify the VA to soft output Viterbi algorithm (SOVA), which generates the soft-decisions by considering paths that merge with the ML path in the trellis within a decision delay. However, such a decision delay is usually quite long, a typical value is $5(L\!+\!1)$, with $L$ being the tap-length of the considered channel impulse response (CIR) $\vec{h}$. In \cite{KB90}, Koch and Baier proposed the soft output Viterbi equalizer (SOVE). Rather than minimizing the sequence error probability in SOVA, the SOVE uses a trellis-based algorithm that minimizes the bit error probability.

To further reduce the receiver complexity, the authors in \cite{KB90} also proposed the suboptimal reduced-state SOVE (RS-SOVE). Different from the SOVE whose trellis spans over all $L$ taps of $\vec{h}$, the trellis in RS-SOVE only spans the first ($\nu\!+\!1$) taps, and the signal part corresponding to the remaining ($L\!-\!\nu\!-\!1$) channel tails is canceled by a state-dependent decision-feedback along the detection. The RS-SOVE is simple to implement and performs nearly as good as the full-complexity SOVE. Note that, the RS-SOVE can also be reviewed as a soft-output extension of the delayed decision-feedback sequence estimation (DDFSE)\cite{DC89}, which combines VA and the decision-feedback detection to approximate the MLSE. 

On the other hand, in order to transform $\vec{h}$ into a new target response, which renders better performance in conjunction with the RS-SOVE, the channel shorteners are commonly utilized prior to the RS-SOVE. Therefore, due to its low-complexity, simple-implementation and good-performance, the RS-SOVE together with channel shortener is widely used in the receiver design of devices in SC systems. A typical overview of such systems is depicted in Fig. 1. Normally, the channel shortener requires three receiver filters to be designed: a prefilter (the tap-length is up to design), a ($\nu\!+\!1$)-tap target response, and a ($L\!-\!\nu\!-\!1$)-tap feedback filter. 

Traditionally, there are two types of processing schemes for designing the channel shortener, namely, the Forney detection model\cite{HR151} which assumes white noise, and the Ungerboeck detection model\cite{RP12} which assumes that the noise is colored according to the target response autocorrelation. A conventional design of the Forney model based channel shortener is to use an all-phase filter to transform $\vec{h}$ into the minimum-phase equivalent $\tilde{\vec{h}}$. Then, the target response is set to the first ($\nu\!+\!1$) taps of $\tilde{\vec{h}}$, while the feedback filter is set to the remaining taps. The all-pass prefilter can be designed based on various criteria \cite{OS89, GOMB02, A01, BZKWH12} such as linear minimum-mean-square-error (LMMSE), linear prediction, and homomorphic filtering. The authors in \cite{BZKWH12} showed that, the homomorphic filter has lower-complexity, simpler hardware-implementation, and superior performance than the other prefilter designs. We refer to such a conventional channel shortener design as the \lq\lq{}HOM\rq\rq{} shortener.

In \cite{RP12}, the Ungerboeck model based channel shortener design was developed. A prefilter $\vec{v}$ and target response $\vec{g}$ are designed to maximize a mutual information lower bound (MILB) based on a mismatched detection model. However, the feedback filter is not utilized in the detection model, which means that the $(L\!-\!\nu\!-\!1)$ channel tails are truncated directly. We refer to such a state-of-the-art design as the \lq\lq{}UBM\rq\rq{} shortener. As there is no feedback filter, with the UBM shortener there is no decision-feedback process in the RS-SOVE. In \cite{H16, KZWR15}, the UBM shortener was successfully implemented for GSM/EDGE systems, and showed superior detection performance, yet with a much lower complexity than the HOM shortener. However, as shown in \cite{KZWR15}, in the high signal-to-noise (SNR) regime\footnote{In relation to higher-order modulations and code-rates, which require high SNRs to decode.}, the UBM shortener suffers from performance losses and renders a bit-error-rate (BER) error floor. 

In this paper, we propose a novel channel shortener design for RS-SOVE aiming to overcome the performance losses of the UBM shortener. As will be explained later, the UBM shortener can not be extended by decision-feedback using the methods introduced in \cite{RP12, H16, KZWR15}. Instead we show that we can overcome the performance losses of the UBM by applying the information theoretical MILB approach to the Forney model instead of the Ungerboeck model. Since we derive a Forney model equalizer that is equipped with MILB-maximization channel shortening filters, we refer to this approach as the \lq\lq{}FOM\rq\rq{} shortener. Note that, both the HOM and FOM shorteners adopt the same the Fornery model for channel shortener designs. The difference is that, the HOM shortener is a conventional design, while the FOM shortener optimizes the receiver filters to maximize the information rate\footnote{The information rate is a bound on the rate that can be transmitted, but are not a capacity since there are constraints on the transmit signals and the decoding operations.}. Therefore, the FOM shortener always performs better than the HOM shortener from an information-theoretical perspective.

On the other hand, if we constrain the feedback filter to be $\vec{0}$, in which case the RS-SOVE utilizes no feedback, the UBM shortener is superior to the FOM shortener. This is because, with no feedback utilized, the Ungerboeck model is more general than the Fornery model. However, when the feedback filter is not $\vec{0}$, the UBM shortener is not applicable for the RS-SOVE, due to the lack of a probabilistic meaning of the branch metric definition\cite{LR08, F15}. Hence, the UBM shortener is constrained to the case that the feedback filter equals $\vec{0}$, while the FOM shortener can jointly optimize all three receiver filters. Therefore, the FOM shortener is superior to the UBM shortener when the feedback has good quality. 

In this work, we show that although at low code-rates the UBM shortener performs better than both the HOM and FOM shorteners, it suffers from significant performance losses at medium and high code-rates. This phenomenon, however, does not exist for the FOM shortener, which outperforms the UBM shortener at medium and high code-rates, and better than the conventional HOM shortener in all cases. These three different channel shorteners considered in this paper are listed in Table I, with FOM shortener being the proposed channel shortener design and the remaining two are the reference designs. 

\begin{table}[ht]
\renewcommand{\arraystretch}{1.2}
\vspace{-3mm}
\centering
\caption{Channel shortener designs and Parameter Notations}
\label{tab1}
\vspace{-2mm}
\begin{tabular}{|c|c|c|c|c|}
\hline
Name & prefilter &\scell{target \\response}& \scell{feedback \\filter} & \scell{RS-SOVE cooperates\\ with feedback?}\\ \hhline{|=|=|=|=|=|} 
FOM&     $\vec{w}$&     $\vec{f}$&  $\vec{b}$& yes \\ \hline
UBM&     $\vec{v}$&      $\vec{g}$&  $\vec{0}$& no   \\ \hline
HOM&    $\vec{w}_{\mathrm{hom}}$ &  $\vec{h}_{\mathrm{f}}$& $\vec{h}_{\mathrm{b}}$& yes \\ \hline
\end{tabular}
\vspace{-2mm}
\end{table}

The main contributions of this paper are as follows. Firstly, we propose the FOM shortener for RS-SOVE with the MILB derived in closed form. The prefilter and feedback filter are found in closed forms, and the target response utilizes a gradient-ascending optimization. Secondly, we analyze the optimal parameter design of the FOM channel shortener by considering the feedback quality, and show that the FOM shortener can be designed for the perfect feedback. We further show that, the FOM shortener outperforms the UBM shortener at medium and high code-rates, and is superior to the HOM shortener in all cases. Lastly, we analyze information-theoretic properties and information rates of the FOM shortener in relation to Shannon capacity $\mathcal{C}$ and the previous channel shortener designs. In addition, we extend the RS-SOVE to an arbitrary delay $D$, and show an interesting fact that, the trellis search process in RS-SOVE is equivalent to a full forward recursion and $D$-depth backward recursion. 

The rest of the paper is organized as follows. In Sec. II, the received signal model, conventional HOM shortener, and RS-SOVE are introduced. In Sec. III, the proposed FOM shortener is derived, and the optimal design of the filters ($\vec{w}$, $\vec{f}$, $\vec{b}$) with feedback quality is elaborated. In Sec. IV, the links of theoretical information rates among all three channel shorteners are established. Empirical results are provided in Sec. V, and Sec. VI concludes the paper.

\subsubsection*{Notations}
Throughout this paper, boldface lowercase letters indicate vectors and boldface uppercase letters designate matrices. Superscripts $(\cdot)^{-1}$, $(\cdot)^{\ast}$,
$(\cdot)\rmt$ and $(\cdot)\rmh$ stand for the inverse, complex conjugate, transpose, and
Hermitian transpose, respectively. Furthermore, $\mathbb{E}[\cdot]$ is the expectation operator, and $\mathcal{R}\{\cdot\}$ takes the real part of the arguments. We reserve \lq{}$\star$\rq{} to denote linear convolution, $\vec{I}$ to represent an identity matrix, and $\mathrm{vec}\left(\vec{A}\right)$ to stack the columns of $\vec{A}$ on top of each other.

\section{Received Signal Model and the HOM Detector}
The considered SC system that applies channel shortening and RS-SOVE is depicted in Fig. 1. With sufficiently good interleaving, we assume the transmit bits to be independent. The transmit symbols $x_k$ have unit-energy and are drawn from a constellation $\mathcal{X}$, whose cardinality is $|\mathcal{X}|$. Considering the data transmission over a dispersive channel with additive noise, the received sample $y_k$ at time epoch $k$ is modeled as
\bea \label{model1} y_k =\sum_{\ell=0}^{L\!-1}h_{\ell}x_{k-\ell}+n_{k},\eea
where $L$ is the ISI duration, and $h_{\ell}$ is the $\ell$th tap of the CIR $\vec{h}\!=\!\left(h_0 \; h_1 \;\ldots \;h_{L-1}\right)$. The noise variables $n_{k}$ are identical and independently distributed (IID) zero-mean complex Gaussian variables with variance $N_0$. For a transmit block comprising $K$ symbols\footnote{We assume that $L\!-\!1$ zero-symbols are inserted between continuous data blocks (i.e., the guard period) to prevent inter-block interference.}, we denote the signal vector $\vec{y}$, receive vector $\vec{x}$, and noise vector $\vec{n}$ as
{\setlength\arraycolsep{2pt}  \bea 
\vec{y}&=&\left(y_0 \; y_{1} \;\ldots \;y_{K+L-1}\right)\rmt, \notag \\
\vec{x}&=&\left(x_0 \; x_{1} \;\ldots \;x_{K-1}\right)\rmt, \notag \\
\vec{n}&=&\left(n_0 \; n_{1} \;\ldots \;n_{K-1}\right)\rmt,\notag\eea}
\hspace{-1.4mm}respectively. The signal model (\ref{model1}) that comprises $(K\!+\!L\!-\!1)$ received samples can  be written as
\bea \label{model2} \vec{y} =\vec{h}\star\vec{x}+\vec{n}, \eea
or equivalently,
  \bea \label{model3}   \vec{y}=\vec{H}\vec{x}+\vec{n},\eea
where the $(K\!+\!L\!-\!1)\!\times\!K$ Toeplitz matrix $\vec{H}$ is generated from $\vec{h}$ as
{\setlength\arraycolsep{2pt} \bea \label{H1} \vec{H}\!=\!\left[\!\begin{array}{cccc} h_0& ~&~&~ \\ h_1&h_0&~&~ \\ \vdots& h_1&\ddots&~\\  h_{L-1}&\vdots&\ddots&h_0 \\ ~&h_{L-1}&\ddots&h_1 \\   ~&~&\ddots&\vdots \\  ~&~&~&h_{L-1}\end{array} \!\right]\!.  \eea}

\begin{figure*}
\hspace{0mm}
\centering
\vspace{-4mm}
\setlength\fboxsep{0pt}
\setlength\fboxrule{0pt}
\fbox{\includegraphics[width=6.6in]{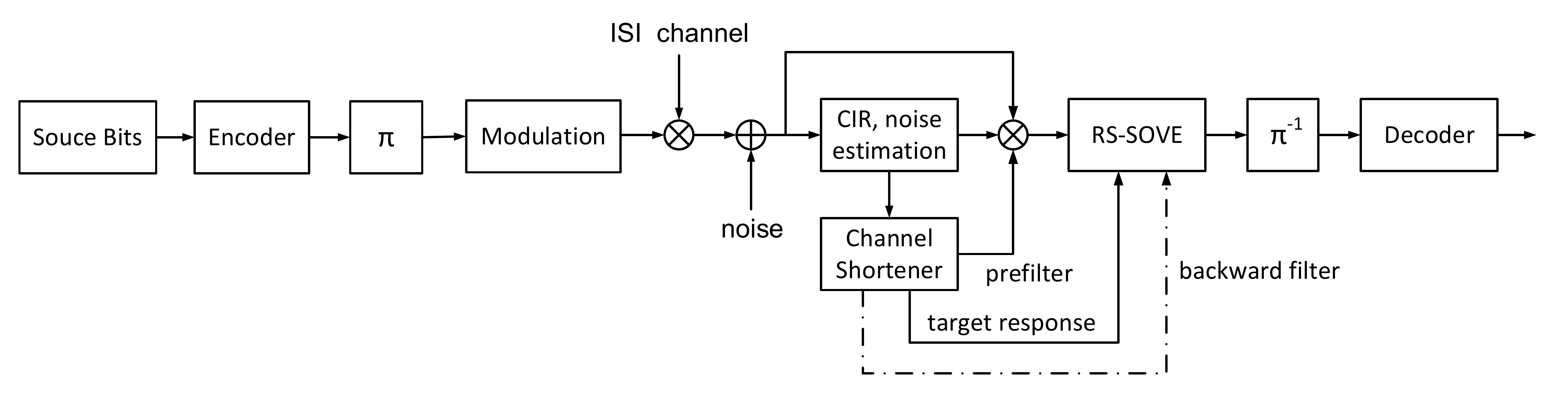}}
\vspace{-15mm}
\caption{Discrete time transmission and receive model with the channel shortener and RS-SOVE. Note that, with the UBM shortener the feedback filter is not needed and no decision-feedback is performed in the RS-SOVE. The CIR and noise estimation can be based on, e.g., pilot symbols.}
\label{fig1}
\vspace*{-6mm}
\end{figure*}

\subsection{Conventional HOM Channel Shortener}
Prior to the RS-SOVE, the HOM shortener utilizes homomorphic filtering to obtain the minimum-phase equivalent form of the causal response $\vec{h}$. With the prefilter $\vec{w}_{\mathrm{hom}}$ designed based on the cepstrum of $\vec{h}$ \cite{GOMB02}, the target response $\tilde{\vec{h}}\!=\!\vec{w}_{\mathrm{hom}}\star(\vec{h}/\sqrt{N_0})$, and the filtered samples $\tilde{\vec{y}}\!=\!\vec{w}_{\mathrm{hom}}\star(\vec{y}/\sqrt{N_0})$, the detection model after prefiltering reads
 \bea \label{model3}  \tilde{y}_k =\sum_{\ell=0}^{\nu}\tilde{h}_{\ell}x_{k-\ell}+\sum_{\ell=\nu+1}^{L-1}\tilde{h}_{\ell}x_{k-\ell}+\tilde{n}_{k},\eea
where $\nu$ denotes the memory length considered by the RS-SOVE so that its number of states becomes $|\mathcal{X}|^{\nu}$. Denoting
{\setlength\arraycolsep{2pt} \bea \label{hf} \vec{h}_{\mathrm{f}} &=&\big(\tilde{h}_0, \tilde{h}_1,\cdots, \tilde{h}_{\nu}\big),\\
\label{hb} \vec{h}_{\mathrm{b}}&=&\big(\underbrace{0, \cdots,0}_{\nu+1}, \tilde{h}_{\nu+1}, \tilde{h}_{\nu+2}, \cdots, \tilde{h}_{L-1}\big),\eea}
\hspace{-1.4mm}the second term in (\ref{model3}) is canceled by the hard feedback $\hat{x}_{\ell}$ on the surviving path that leads to each state after filtered by the feedback filter $\vec{h}_{\mathrm{b}}$. By setting $\nu\!=\!0$, the RS-SOVE becomes the decision-feedback detector, while with $\nu\!=\!L\!-\!1$, the RS-SOVE is the full-complexity SOVE. 

In contrast to BCJR algorithm \cite{bcjr} or Max-Log-Map (MLM)\cite{M05}, the backward recursions are omitted in RS-SOVE \cite{KB90}. In order to improve the quality of soft-decisions, we extend the decision-delay in RS-SOVE to an arbitrary value $D$, which can set to be larger than $\!L\!-\!1$. As shown next, the RS-SOVE with a delay $D$ can be viewed as the the MLM equalizer with a full forward recursion and $D$-step backward recursion at each detection stage. Hence, when $D$ is sufficiently large, the RS-SOVE performs as well as MLM. Such a modification only increases the equalization latency from $\nu$ to $D$, and introduces a small overheard by the $D$-step backward recursion process in the RS-SOVE. In \cite{BF98}, an improvement of RS-SOVE is also proposed by introducing an expanded memory, however, the number of states is exponentially increased and results in higher memory cost.

\begin{figure}
\begin{center}
\hspace*{0mm}
\vspace*{-4mm}
\scalebox{0.48}{\includegraphics{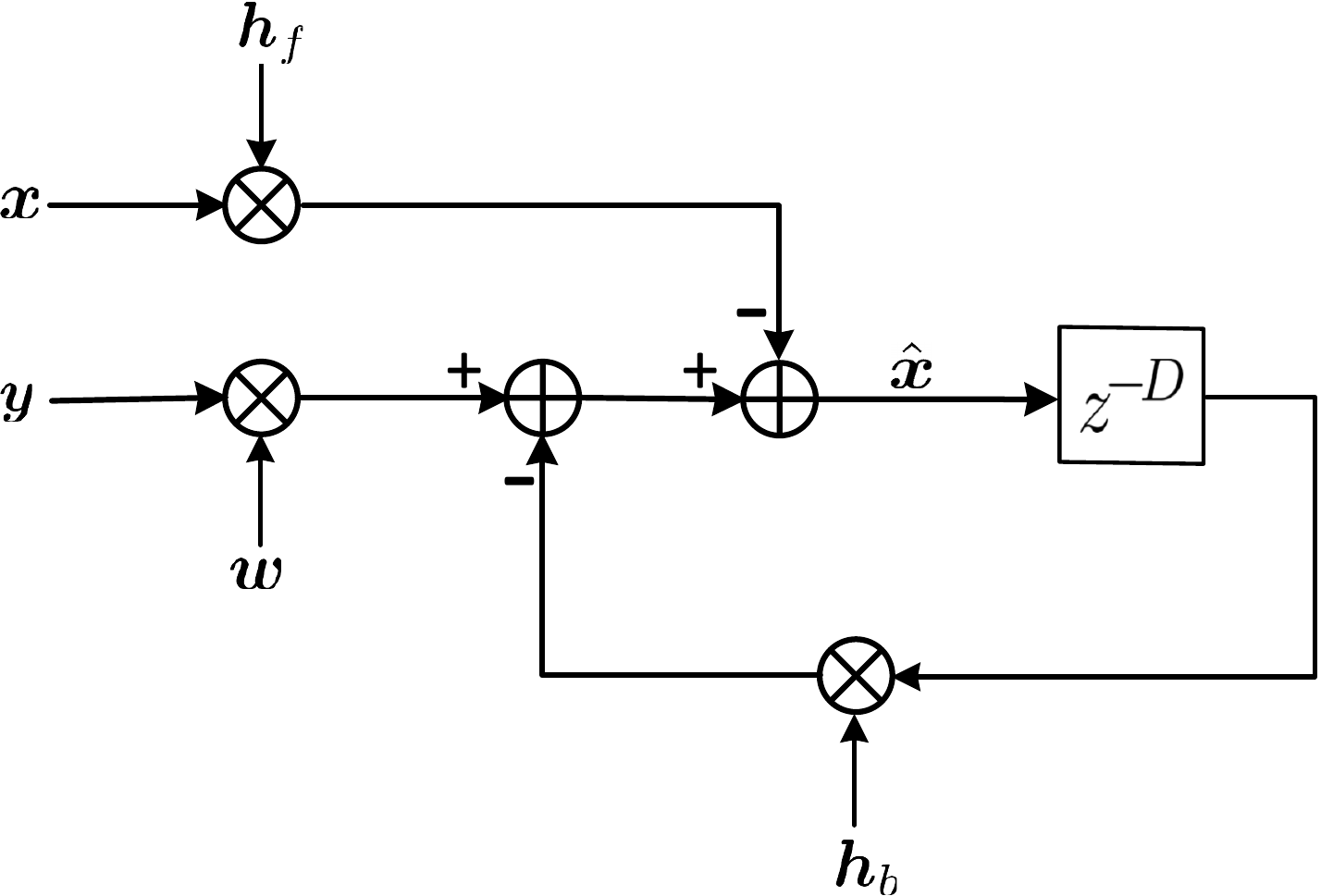}}
\vspace*{-2mm}
\caption{\label{fig2}The decision-feedback process in RS-SOVE. The hard feedback $\hat{\vec{x}}$ is associated to each state and updated along the detection stages.}
\end{center}
\vspace*{-12mm}
\end{figure}

\begin{figure*}[b]
\vspace{-6mm}
\hrulefill
\vspace{1mm}
{\setlength\arraycolsep{0pt} \bea \label{llr}  L(x_{k,n})&=&  \log\!\left( \sum_{x_{k,n}=1}\!\exp\left( \!-\alpha_{k+\nu-1}^i- \gamma_{k+\nu}^{i,j}-\beta_{k+\nu}^j\right)\!\right)\!\!- \log\!\left( \sum_{x_{k,n}=-1} \!\exp\left(\!-\alpha_{k+\nu-1}^i- \gamma_{k+\nu}^{i,j}-\beta_{k+\nu}^j\right)\!\right) \!\!\! \! \! \! \! \! \! \!  \notag\\
&\approx&\min_{x_{k,n}=-1}\left( \alpha_{k+\nu-1}^i+ \gamma_{k+\nu}^{i,j}+\beta_{k+\nu}^j\right)-\min_{x_{k,n}=1}\left(\alpha_{k+\nu-1}^i+ \gamma_{k+\nu}^{i,j}+\beta_{k+\nu}^j\right) .\eea}
\vspace{-6mm}
\end{figure*}

\subsection{RS-SOVE with Arbitrary Decision-Delay $D$}
In Fig. \ref{fig2}, we illustrate the decision-feedback detection in the RS-SOVE with the prefilter $\vec{w}$, the feedback filter $\vec{h}_{\mathrm{b}}$, and the target response $\vec{h}_{\mathrm{f}}$. Utilizing Jacobian approximation \cite{M05}, $$\log\big(\exp(-a)+\exp(-b)\big)\!\approx\!-\min(a,b),$$ the soft-decisions of the $n$th bit $x_{k,n}$ in $x_{k}$, i.e., the log-likelihood ratio (LLR), is calculated according to (\ref{llr}) with delay $\nu$. The forward path metric $\alpha_{k}^j$ corresponding to state $j$ at stage $k$ is recursively computed through
\bea \label{alpha} \alpha_{k}^j=\min_{i}\left\{ \alpha_{k-1}^i+\gamma_k^{i,j}\right\},\eea
where the branch metric $\gamma_k^{i,j}$ in (\ref{model3}) associated to state transition $i\!\to\! j$ is calculated as
 \bea \label{metric1} \gamma_k^{i,j}=\left|\tilde{y}_k-\sum_{\ell=0}^{\nu}\tilde{h}_{\ell} x_{k-\ell}-\sum_{\ell=\nu+1}^{L\!-1}\tilde{h}_{\ell} \hat{x}_{k-\ell}\right|^2.\eea
In (\ref{metric1}), the symbol vector $(x_k,\cdots,x_{k-\nu})$ are determined from state transition $i\!\to\! j$, while $(\hat{x}_{k-\nu-1},\cdots,\hat{x}_{k-L+1})$ are the hard decisions associated to each state $i$ at stage $k$. As for each state there is a survival path that leads to it, with decision-feedback determined from such a path, the feedback varies on different states. In addition, an update of all survival paths is needed along the detection stages.

In \cite{KB90}, with RS-SOVE the backward recursions are omitted by setting $\beta_{k+\nu}^j\!=\!0$ for all states, and the LLR in (\ref{llr}) is simplified to
 \bea  L(x_{k,n})\!\approx\!\min_{x_{k,n}=-1}\!\left(\! \alpha_{k+\nu-1}^i\!+\! \gamma_{k+\nu}^{i,j}\!\right)\!-\!\min_{x_{k,n}=1}\!\left(\!\alpha_{k+\nu-1}^i\!+\! \gamma_{k+\nu}^{i,j}\!\right)\!.\notag \eea
However, a drawback of such an approximation is that, the short decision delay $\nu$ in RS-SOVE limits its performance, especially with higher-order modulations and code-rates\cite{BF98}. Therefore, we increase the delay $\nu$ to an arbitrary value $D$ by initializing $\beta_{k+D}^j\!=\!0$ for all states at detection stage $k\!+\!D$, and define the backward recursion for state transition $j\!\to\! i$ as
 \bea \label{beta} \beta_{k-1}^i=\min_{j}\left\{ \beta_{k}^j+\gamma_k^{i,j}\right\}. \eea
Note that, from detection stage $k$ up to $k\!+\!\nu\!-\!1$, the state transactions corresponding to different symbol assumptions $x_k$ do not merge with each other at the same state (and on both directions). This is so, since the state transactions from stage $k$ to $k\!+\!\nu\!-\!1$ follow the below pattern $$\underbrace{\mathrm{o}\underbrace{\mathrm{xx}\cdots\mathrm{x}}_{\nu-1}\longrightarrow \mathrm{xo}\underbrace{\mathrm{xx}\cdots\mathrm{x}}_{\nu-2}\longrightarrow \cdots \underbrace{\mathrm{xx}\cdots\mathrm{x}}_{\nu-1}\mathrm{o}}_{\nu\;\mathrm{stages}},$$
where \lq\lq{}$\mathrm{o}$\rq\rq{} denotes the symbol assumption $x_k$ at stage $k$, and \lq\lq{}$\mathrm{x}$\rq\rq{} represents all the possible choices for the other $\nu\!-\!1$ symbols on each state. There all in total $|\mathcal{X}|$ possible assumptions for $x_k$, and with each assumption, the sub-trellises formed by the transition pattern above are non-intersecting within stage $k$ and $k\!+\!\nu\!-\!1$. Hence, by utilizing (\ref{alpha}) and (\ref{beta}), the minimal path metric of each symbol assumption $x_k$ in (\ref{llr}) can be recursively computed as
{\setlength\arraycolsep{2pt} \bea \label{llr2}  \min_{x_{k}}\left( \alpha_{k+\nu-1}^i+ \gamma_{k+\nu}^{i,j}+\beta_{k+\nu}^j\right)&=&\min_{x_{k}}\left( \alpha_{k+\nu-1}^i+ 
\beta_{k+\nu-1}^i\right) \notag \\
&=&\min_{x_{k}}\left( \alpha_{k}^i+\beta_{k}^i\right).\eea}
\hspace{-1.4mm}Then, for each bit assumption $x_{k,n}$ the minimal path metric is the minimum of all $|\mathcal{X}|/2$ symbols $x_k\in\mathcal{X}$ that the $n$th bit equals to such an assumption. Therefore, the LLR in (\ref{llr}) can be equivalently expressed as
{\setlength\arraycolsep{2pt} \bea \label{llr1}  L(x_{k,n})=\min_{x_{k,n}=-1}\left( \alpha_{k}^j+\beta_{k}^j\right)-\min_{x_{k,n}=1}\left(\alpha_{k}^j+\beta_{k}^j\right). \eea

In Fig. \ref{fig3}, we illustrate the forward and backward recursions in the RS-SOVE at detection stage $k$ with a binary trellis with $\nu\!=\!2$ and $D\!=\!4$. As can be seen, the state transactions represented by the red lines and blues lines (both solid and dashed lines) do not merge with each other at stage $k$ and $k\!+\!1$, and the recursion (\ref{llr2}) holds. The LLR calculation in (\ref{llr1}) shows that, with an arbitrary delay $D$ and branch metric computation in (\ref{metric1}), the RS-SOVE can be reviewed as an MLM equalizer, but with a full forward recursion and $D$-step back recursion at each stage. 

Next, we introduce the proposed optimal FOM shortener design that cooperates with decision-feedback in the RS-SOVE which has been introduced in this section.

\begin{figure*}[b]
\vspace{-2mm}
\centering
\hspace{-1mm}
\setlength\fboxsep{0pt}
\setlength\fboxrule{0pt}
\scalebox{0.36}{\includegraphics{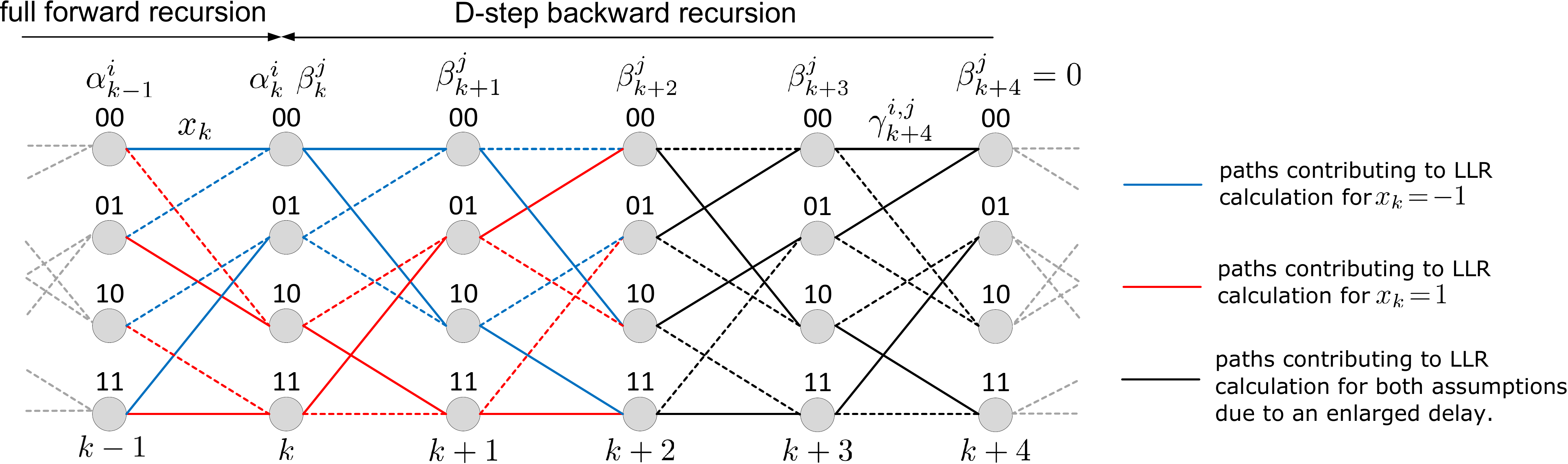}}
\vspace{-11mm}
\caption{A trellis diagram for binary-phase-shift-keying (BPSK) modulation with memory length $\nu\!=\!2$ and an enlarged delay $D\!=\!4$. In RS-SOVE, the decision-feedback is determined by the survival path on each state. The last ($L\!-\!\nu\!-\!1$) symbols associated to the survival path that leading to current state are preserved and updated along the detection stages. The dashed lines are the discarded paths in the forward and backward recursions due to the Jacobian approximation.}
\label{fig3}
\vspace{-4mm}
\end{figure*}

\section{The Optimal FOM Channel Shortener Design for RS-SOVE}
As the HOM shortener is a static and heuristic approach, it neither takes the noise power nor the quality of feedback $\hat{\vec{x}}$ into account when designing $\vec{w}_{\mathrm{hom}}$. Consequently, the detection performance is often inferior to the UBM shortener\cite{KZWR15}. Moreover, the UBM shortener also suffers from performance losses in middle and high SNR regimes. The reason is that, as mentioned earlier, the channel tails are truncated and the RS-SOVE does not cooperate with feedback. On the other hand, with high SNR the hard decisions are sufficiently good along the ML path, which can be exploited to cancel the signal part corresponding to the channel tails, instead of direct truncating which ends up with a transmission-energy loss.

Since we are dealing with ISI channels, the FOM receiver filters are designed assuming a large $K$, in which case we can let $\vec{H}$ represent the $K\!\times\!K$ circular convolution matrix instead of the normal convolution\footnote{Another conceptually simple way to interpret this is to replace the first $L\!-\!1$ symbols in $\vec{x}$ with its last $L\!-\!1$ symbols, i.e., inserting CP. But, here we make such an approximation on $\vec{H}$ is solely for the sake of designing optimal parameters of the channel shortener. We do not insert CP in the transmit blocks when evaluating the detection performance later.}. Such an approximation has no impact on the information rate as $K\!\to\!\infty$, see e.g., \cite{W88} for a rigorous information-theoretic treatment. From Szeg\"{o}'s eigenvalue distribution theorem\cite{Sze,MK}, the eigenvalues of Toeplitz matrices converge to the Fourier transforms of the sequences that they induce. This implies that, we can equivalently work with the Fourier transforms of all involved Toeplitz matrices, or the vectors that specify them. 

Denote the discrete-time Fourier transform (DTFT) of vector $\vec{h}$ and the inverse operation (IDTFT) as
{\setlength\arraycolsep{2pt} \bea \label{Hw}  H(\omega) &=&\sum_{\ell=0}^{L-1}h_{\ell}\exp(j\omega\ell), \\
  \label{idtft} h_{\ell}&=& \frac{1}{2\pi}\int_{-\pi}^{\pi} H(\omega)\exp(-j\omega\ell) \mathrm{d}\omega. \eea}
\hspace{-1.4mm}respectively. Next, we elaborate the optimal FOM shortener design. Although we adopt the same approach as MILB-maximization, the FOM shortener is different from the previous designs\cite{RP12, KZWR15}, which are based on Ungerbeock model and take no feedback into consideration. In \cite{FDG12}, the authors extend the UBM shortener to deal with soft feedback and with turbo iterations. However, with RS-SOVE, there are no turbo iterations and the UBM shortener is not applicable.

\subsection{The FOM Channel Shortener Design with Feedback}
Consider the Forney detection model with feedback,
  \bea  \label{md3} \tilde{p}(\vec{y}|\vec{x},\hat{\vec{x}})=\exp\big(-\left\|\vec{W}\vec{y}-\vec{F}\vec{x}-\vec{B}\hat{\vec{x}}\right\|^2\big), \eea
where $\vec{W}$, $\vec{F}$ and $\vec{B}$ are $K\!\times\! K$ convolution matrices generated from $\vec{w}$, $\vec{f}$ and $\vec{b}$, respectively, and $\hat{\vec{x}}$ is the feedback. There is no constraint on $\vec{w}$, and $\vec{f}$, $\vec{b}$ are as below\footnote{Although with arbitrary $\vec{w}$, the feedback filter $\vec{b}$ can be arbitrary long, we make such constraints to align the complexity of decision-feedback detection in the RS-SOVE corresponding to the HOM shortener.},
{\setlength\arraycolsep{2pt} \bea \label{f}  \vec{f} &=&\big(f_0,f_1,\cdots, f_{\nu}\big),\\
\label{b} \vec{b} &=&\big(\underbrace{0, \cdots,0}_{\nu+1}, b_{0}, b_{1}, \cdots, b_{L-\nu-2}\big). \eea}
\hspace{-1.4mm}The receiver filters ($\vec{w}, \vec{f}, \vec{b}$) are optimized through maximizing the MILB, which is defined as
\be \label{metric} I_{\mathrm{LB}} = \lim_{K\to\infty}\frac{1}{K}\Big(\mathbb{E}_{\vec{x},\vec{y}}\big[\ln \tilde{p}(\vec{y}|\vec{x},\hat{\vec{x}})\big]-\mathbb{E}_{\vec{y}}\big[\ln\tilde{p}(\vec{y}|\hat{\vec{x}})\big]\Big),\ee
where the expectations are taken over the true channel statistics\footnote{In order to obtain a tractable problem\cite{RP12}, we make the assumption that $\vec{x}$ comprises IID complex Gaussian variables when calculating $I_{\mathrm{LB}}$.} expression of $I_{\mathrm{LB}}$ and with $\tilde{p}(\vec{y}|\vec{x},\hat{\vec{x}})$ in (\ref{md3}),
\bea \tilde{p}(\vec{y}|\hat{\vec{x}})=\int_{\vec{x}}\tilde{p}(\vec{y}|\vec{x},\hat{\vec{x}})p(\vec{x})\mathrm{d}\vec{x}. \notag \eea

The quality of feedback $\hat{\vec{x}}$, which impacts the rate $I_{\mathrm{LB}}$, is measured by two parameters, 
{\setlength\arraycolsep{2pt} \bea \eta&=&\frac{1}{K}\mathbb{E}[\hat{\vec{x}}\hat{\vec{x}}^\dag],  \notag \\
\label{sigma} \sigma&=& \frac{1}{K}\mathbb{E}[\hat{\vec{x}}\vec{x}^\dag]. \eea }
\hspace{-1.8mm}In RS-SOVE, $\hat{\vec{x}}$ are hard symbols and we have $\eta\!=\!1$. With soft symbols feedback, $\eta$ can be calculated from the variance of the estimates, i.e., $\eta\!=\!1-\mathrm{var}({\hat{\vec{x}}})$. With optimal ($\vec{w}$, $\vec{f}$, $\vec{b}$), and denoting $\tilde{\vec{y}}$ as the received samples after filtering by $\vec{w}$, the branch metric $\gamma_k^{i,j}$ in (\ref{metric1}) is calculated as
 \bea \label{metric2} \gamma_k^{i,j}=\left|\tilde{y}_k-\sum_{\ell=0}^{\nu}f_{\ell} x_{k-\ell}-\sum_{\ell=0}^{L\!-\nu\!-2}b_{\ell} \hat{x}_{k-\ell-\nu-1}\right|^2.\eea

Before optimizing ($\vec{w}$, $\vec{f}$, $\vec{b}$), we introduce the following notations. Following (\ref{Hw}), we denote the DTFT of $\vec{w}$, $\vec{f}$, and $\vec{b}$ as $W(\omega)$, $F(\omega)$, and $B(\omega)$, respectively. Then, we let
{\setlength\arraycolsep{2pt} \bea \label{mw} M(\omega)&=&-\frac{N_0}{N_0+|H(\omega)|^2},  \\
 \label{mtw} \tilde{M}(\omega)&=&\sigma^{2}\left(1+M(\omega)\right)-\sigma, \eea}
\hspace*{-1.2mm}and
 \bea \label{phi} \vec{\phi}(\omega)=\big[\exp\left(j\omega(\nu\!+\!1)\right)\; \exp\left(j\omega(\nu\!+\!2)\right)\; \ldots \; \exp\left(j\omega(L\!-\!1)\right)\big]\rmt\!. \quad\eea}

Further, denote $(L\!-\!\nu\!-\!1)\!\times\!1$ vector $\vec{\varepsilon}_1$, and $(L\!-\!\nu\!-\!1)\!\times\!(L\!-\!\nu\!-\!1)$ Hermitian matrix $\vec{\varepsilon}_2$ as
{\setlength\arraycolsep{2pt} \bea \label{vareps1} \vec{\varepsilon}_1&=&\frac{\sigma}{2\pi}\int_{-\pi}^{\pi}M(\omega)F^{\ast}(\omega)\vec{\phi}(\omega)\mathrm{d}\omega,  \\
\label{vareps2} \vec{\varepsilon}_2&=&\frac{1}{2\pi} \int_{-\pi}^{\pi}\frac{\tilde{M}(\omega)|F(\omega)|^2\vec{\phi}(\omega)\vec{\phi}(\omega)\rmh}{1+|F(\omega)|^2}\mathrm{d}\omega.\eea}
\hspace*{-1.4mm}With definitions in (\ref{mw})-(\ref{vareps2}), we have the below lemma that states the closed-form MILB.
\begin{lemma}
The MILB in (\ref{metric}) equals
\bea \label{ibarm1} I_{\mathrm{LB}}\!
\!&=&\!\!\frac{1}{2\pi}\! \int_{-\pi}^{\pi}\!\!\bigg(\! \log \!\big(1\!+\! |F(\omega)|^2\big)\!-\! |F(\omega)|^2 \!-\!\frac{L(\omega)}{1\!+\!|F(\omega)|^2}\! \bigg)\mathrm{d}\omega \nonumber \\
&&\!\!\!+\frac{1}{\pi}\!\int_{-\pi}^{\pi}\mathcal{R}\big\{F^{\ast}(\omega)\big(W(\omega)H(\omega)\!-\!\sigma B(\omega)\big)\big\}\mathrm{d}\omega,
\eea
where
{\setlength\arraycolsep{2pt} \bea L(\omega)&=&|F(\omega)W(\omega)|^{2}\big(N_0+|H(\omega)|^2\big)+\sigma|F(\omega)B(\omega)|^2\notag \\
&&\!-\!  2\sigma|F(\omega)|^{2}\mathcal{R}\!\left\{H(\omega)W(\omega)B^{\ast}(\omega) \right\} \notag.\eea}
\end{lemma}
\begin{proof}
In \cite[eq.(5)-(6)]{HR151}, the generalized mutual information $I_{\mathrm{GMI}}$ is derived for any $K\!\times\!K$ linear multi-input and multi-output (MIMO) channel. For ISI channels, which can be viewed as special cases of MIMO channel, it holds that $I_{\mathrm{LB}}\!=\!\lim\limits_{K\to\infty}\frac{1}{K}I_{\mathrm{GMI}}$. By applying Szeg\"{o}'s theorem and after some manipulations, (\ref{ibarm1}) follows.
\end{proof}

With $I_{\mathrm{LB}}$ stated in (\ref{ibarm1}), the optimal $W(\omega)$ and $B(\omega)$ that maximize $I_{\mathrm{LB}}$ are in Theorem 1.
\begin{theorem} \label{thm1}
The optimal $W(\omega)$ that maximizes $I_{\mathrm{LB}}$ equals,
\bea \label{optw} W_{\mathrm{opt}}(\omega)=\frac{H^{\ast}(\omega)\big(1+ |F(\omega)|^2+\sigma F(\omega)B^{\ast}_{\mathrm{opt}}(\omega)\big)}{F^{\ast}(\omega)(N_0+|H(\omega)|^2)},\eea
and when $\sigma>0$, the optimal $B(\omega)$ reads,
\bea \label{optb} B_{\mathrm{opt}}(\omega)=-\vec{\varepsilon}_1\rmh\vec{\varepsilon}_2^{-\!1}\vec{\phi}(\omega).\eea
With $W_{\mathrm{opt}}(\omega)$ and $B_{\mathrm{opt}}(\omega)$, $I_{\mathrm{LB}}$ equals,
{\setlength\arraycolsep{2pt} \bea \label{Ilb}I_{\mathrm{LB}}=\left\{\begin{array}{ll} \mathcal{J}\big(F(\omega)\big),& \sigma=0,\\
\mathcal{J}\big(F(\omega)\big)-\vec{\varepsilon}_1\rmh\vec{\varepsilon}_2^{-1}\vec{\varepsilon}_1, &0<\sigma\leq1, \end{array}\right.\eea}
\hspace*{-1.4mm}where $\mathcal{J}\big(F(\omega)\big)$ reads
 \bea \label{jw} \mathcal{J}\big(F(\omega)\big)= 1+\frac{1}{2\pi}\int_{-\pi}^{\pi}\!\!\Big(\log\big(1+|F(\omega)|^2\big) +M(\omega)\big(1+|F(\omega)|^2\big) \Big)\mathrm{d}\omega. \; \eea
\end{theorem}
\begin{proof}
See Appendix A.
\end{proof}
In (\ref{Ilb}), the term $-\vec{\varepsilon}_1\rmh\vec{\varepsilon}_2^{-1}\vec{\varepsilon}_1$ is the information rate increment due to the feedback $\hat{\vec{x}}$. From Theorem 1, $W(\omega)$, $B(\omega)$ are in closed forms, and $\vec{w}$, $\vec{b}$ can be obtained through IDTFT operations. But for $F(\omega)$ and $\vec{f}$, a closed form solution can not be reached. Hence, we use a gradient-ascending based optimization, with the updating at each iteration defined as
\bea \label{ga} \vec{f}^{i}=\vec{f}^{i-1}+\nabla_{\vec{f}^{\ast}} I_{\mathrm{LB}}.\eea
As the DTFT of $\vec{f}$ reads
\bea F(\omega)=\sum_{k=0}^{\nu}f_k\exp\!\big(j k\omega \big), \notag\eea
the first-order derivatives of $\mathcal{J}(F(\omega))$ and $\vec{\varepsilon}_1\rmh\vec{\varepsilon}_2^{-1}\vec{\varepsilon}_1$ in (\ref{Ilb}) with respect to $f_k$ read
{\setlength\arraycolsep{1pt}  \bea &&\frac{\partial \mathcal{J}}{\partial f_k}=\frac{1}{2\pi} \!\int_{-\pi}^{\pi}\!\!\bigg(M(\omega)\!+\!\frac{1}{1\!+\!|F(\omega)|^2}\bigg)F^{\ast}(\omega)\exp\!\big(jk\omega \big)\mathrm{d}\omega, \notag \\
 &&\frac{\partial \vec{\varepsilon}_1\rmh\vec{\varepsilon}_2^{-1}\vec{\varepsilon}_1}{\partial f_k}=-\frac{\partial \vec{\varepsilon}_1\rmh}{\partial f_k}\vec{\varepsilon}_2^{\!-\!1}\vec{\varepsilon}_1\!+\! \vec{\varepsilon}_1\rmh\vec{\varepsilon}_2^{\!-\!1}\frac{\partial \vec{\varepsilon}_2}{\partial f_k}\vec{\varepsilon}_2^{\!-\!1}\vec{\varepsilon}_1, \notag
 \eea}
 \hspace*{-1.4mm}respectively, and
 {\setlength\arraycolsep{1pt}
  \bea  \frac{\partial \vec{\varepsilon}_1\rmh}{\partial f_k}&=&\frac{\sigma}{2\pi}\!\int_{-\pi}^{\pi}\!M(\omega)\vec{\phi}(\omega)\rmh\exp\!\left(j k\omega \right)\!\mathrm{d}\omega, \notag \\
  \frac{\partial \vec{\varepsilon}_2}{\partial f_k}&=&\frac{1}{2\pi}\! \int_{-\pi}^{\pi}\!\frac{\tilde{M}(\omega)|F(\omega)|^2\vec{\phi}(\omega)\vec{\phi}(\omega)\rmh}{\big(1\!+\!|F(\omega)|^2\big)^2}F^{\ast}(\omega)\exp\!\big(jk\omega \big)\mathrm{d}\omega. \notag\eea}
\;\;Although due to the non-concaveness of $I_{\mathrm{LB}}$ in (\ref{Ilb}), the optimization may converge to a local maximum, such an optimization over $\vec{f}$ is still meaningful, in the sense that the MILB is increased even with a local maximum attained. We initialize $\vec{f}$ in (\ref{ga}) with $\vec{h}_{\mathrm{f}}$ obtained from the HOM shortener. When $N_0$ decreases and with $\sigma\!=\!1$, such an initialization is asymptotically close to the maximum point as the HOM shortener performs close the the FOM shortener, due to the perfect feedback. 

\subsection{The UBM Channel Shortener Design without Feedback}
Next, we introduce the UBM shortener design. By replacing $\vec{V}\!=\!\vec{F}\rmh\vec{W}$, $\vec{R}\!=\!\vec{F}\rmh\vec{B}$ and $\vec{G}\!=\!\vec{F}\rmh\vec{F}$, the model (\ref{md3}) can be rewritten as
\bea \label{UB1}  \tilde{p} (\vec{y}|\vec{x})=\exp\!\Big(2\mathcal{R}\{\vec{x}\rmh(\vec{V}\vec{y}-\vec{R}\hat{\vec{x}})\}-\vec{x}\rmh\vec{G}\vec{x}+\vec{\vartheta}\Big),
\eea
where $\vec{\vartheta}\!=\!-\|\vec{W}\vec{y}-\vec{B}\hat{\vec{x}}\|^2$. In the design of the UBM shortener, $\vec{G}$ is an arbitrary Hermitian matrix and can be non-positive definite\cite{RP12}. In the RS-SOVE, the term $\vec{\vartheta}$ is calculated with the survival path on each state. In order to calculate $\vec{\vartheta}$, we need to decompose $\vec{G}\!=\!\vec{F}\rmh\vec{F}$, which requires $\vec{G}$ to be positive definite. In such a case, (\ref{UB1}) is identical to (\ref{md3}), that is, the UBM shortener becomes the FOM shortener. This dilemma makes the Ungerboeck model not suitable for decision-feedback detection. But with turbo iterations, as $\hat{\vec{x}}$ is known before the RS-SOVE, it is the same for all states and $\vec{\vartheta}$ can be removed from (\ref{UB1}). However, as we are designing a channel shortener with no turbo iterations, we assume no feedback and (\ref{UB1}) changes to
\bea \label{UB2}   \tilde{p} (\vec{y}|\vec{x})=\exp\!\Big(2\mathcal{R}\{\vec{x}\rmh\vec{V}\vec{y}\}-\vec{x}\rmh\vec{G}\vec{x}\Big).\eea
The $K\!\times\! K$ convolution matrix $\vec{V}$ generated from vector $\vec{v}$ has the same structure as $\vec{W}$, while the $K\!\times\! K$ Toeplitz matrix $\vec{G}$ is Hermitian and band-shaped, with only the middle $2\nu\!+\!1$ diagonals can take non-zero values. Denote the vector comprises the first ($\nu\!+\!1$) elements in the first column of $\vec{G}$ as
\bea \vec{g}=\big(g_0, g_1, \cdots, g_{\nu}\big). \notag \eea 
With optimal ($\vec{v}$, $\vec{g}$), and denoting $\tilde{\vec{y}}$ as the received samples after filtering by $\vec{v}$, the branch metric $\gamma_k^{i,j}$ is calculated as
 \bea \label{metric3}  \gamma_k^{i,j}=g_0 |x_k|^2 -2\mathcal{R}\bigg\{x_k^{\ast}\Big(\tilde{y}_k-\sum_{\ell=1}^{\nu}g_{\ell} x_{k-\ell}\Big)\bigg\}.\eea
The model (\ref{UB2}) has been considered in earlier literatures such as \cite{RP12, H16}. The optimal solutions of ($\vec{v}$, $\vec{g}$) can be found in \cite{RP12}, which can also be deduced from Theorem 1 directly. By setting $\sigma\!=\!0$, the optimal $V(\omega)$ for (\ref{UB2}) reads
\bea  \label{optisiv} V_{\mathrm{opt}}(\omega)=\frac{H^{\ast}(\omega)}{N_0+|H(\omega)|^2}\big(1+G(\omega)\big), \eea
and the optimal $G(\omega)$ is the unique solution that maximizes $I_{\mathrm{LB}}$ in (\ref{metric}), which is evaluated based on (\ref{UB2}) and equals
 \be \label{Ilb2} I_{\mathrm{LB}}= 1+\frac{1}{2\pi}\int_{-\pi}^{\pi}\Big(\log\big(1+G(\omega)\big)+M(\omega)\big(1+G(\omega)\big) \! \Big)\mathrm{d}\omega.  \ee

Comparing (\ref{Ilb2}) to (\ref{jw}), the only difference is that $|F(\omega)|^2$ in (\ref{jw}) is replaced by $G(\omega)$. Therefore, the UBM shortener is more general than the FOM shortener under the case that $\sigma\!=\!0$. We point out the fact that, both the FOM and UBM shorteners are invariant under the minimum-phase transforming of the original channel $\vec{h}$. This is because, the homomorphic filter $\vec{w}_{\mathrm{hom}}$ is an all-pass filter, which has no impact on the noise statistical properties, and then the convolution matrix generated from the all-pass filter will be absorbed by the prefilters $\vec{W}$ and $\vec{V}$, respectively. Hence, with the FOM and UBM shorteners, it is no need to transform $\vec{h}$ into an minimum-phase equivalent form prior to prefiltering.

\subsection{Design the Optimal $\sigma$ for the FOM Channel Shortener}
In Theorem 1, the optimal $W(\omega)$ and $F(\omega)$ are related to the feedback quality parameter $\sigma$. However, according to the expectation in (\ref{sigma}), $\sigma$ is hard to find at the design stage. Moreover, it is not necessarily optimal to use the $\sigma$ calculated with (\ref{sigma}). Therefore, it is a free optimization parameter. In the next, we analyze the optimal design of $\sigma$. 

With higher-order modulations, we assume that when a symbol error occurs on the ML path, the hard decision $\hat{x}_k$ and the transmit symbol $x_k$ are independent. Then,
{\setlength\arraycolsep{2pt}  \bea \label{sigma1} \sigma&=&(1-P_{\mathrm{e}})\cdot\mathbb{E}[|x_k|^2]+P_{\mathrm{e}}\cdot\mathbb{E}[\hat{x}_k x_k^{\ast}],  \notag \\
&\approx&1-P_{\mathrm{e}},\eea}
\hspace{-2mm}where $P_{\mathrm{e}}$ is the symbol error rate (SER) of the RS-SOVE. As LMMSE detection is a special case of MILB detection with $\nu\!=\!0$, when $\nu\!>\!0$, the FOM shortener with the RS-SOVE outperforms the LMMSE detector and renders a lower SER \cite{HR151}. That is, denoting $\hat{\vec{x}}^{\mathrm{LMMSE}}$ as the LMMSE estimate and $P_{\mathrm{e}}^{\mathrm{LMMSE}}$ as the corresponding SER, it holds that
{\setlength\arraycolsep{2pt}   \bea  \label{pe}  P_{\mathrm{e}}&\leq& P_{\mathrm{e}}^{\mathrm{LMMSE}} \notag \\
&\stackrel{(a)}{\leq}& \label{sigma0} \mathbb{E}\left[\left(\vec{x}-\hat{\vec{x}}^{\mathrm{LMMSE}}\right)\left(\vec{x}-\hat{\vec{x}}^{\mathrm{LMMSE}}\right)^\dag\right] \Big/2\notag  \\
&\stackrel{(b)}{=}&\delta_{\mathrm{mse}}/2.  \eea}
\hspace{-1.2mm}where
\bea \delta_{\mathrm{mse}}=-\frac{1}{2\pi}\int_{-\pi}^{\pi}M(\omega)\mathrm{d}\omega,  \notag\eea 
which is the MSE of the LMMSE estimate. The inequality (a) is proved in Appendix B, and the equality (b) is from Szeg\"{o}'s eigenvalue distribution theorem. Hence, from (\ref{sigma1}) and (\ref{pe}),
\bea \label{sigma2} \sigma\geq 1-\delta_{\mathrm{mse}}/2.\eea 
The inequality (\ref{sigma2}) provides some insight about designing $\sigma$ for the FOM shortener. As we are expecting that, the RS-SOVE with decision-feedback shall outperform itself without feedback, i.e., $\sigma\!=\!0$, the input $\sigma$ to design $(\vec{w}, \vec{f}, \vec{b})$ should be set to, at least larger than $1\!-\!\delta_{\mathrm{mse}}/2$. Hence, when $\delta_{\mathrm{mse}}$ is small, we can let $\sigma\!=\!1$. 

As we will show next through empirical results, the optimal $\sigma$ can be chosen as either 0 or 1. The reason behind this phenomenon is that, when SNR is low, implying that the quality of $\hat{\vec{x}}$ in the RS-SOVE is fairly poor, it is better to truncate the channel tails to prevent error-propagation. However, when SNR is above a certain threshold, the feedback quality improves and the FOM shortener will benefit from $\hat{\vec{x}}$, in which case we can let $\sigma\!=\!1$. The test set-up for exploiting the relationships between the input and output $\sigma$ is depicted in Fig. {\ref{fig4}. We use Monte Carlo simulations under the below two standard ISI channels.

\begin{example}EPR-4\cite{epr4}. The 4-tap ISI channel, $\vec{h}\!=\![\;0.5\;\; 0.5\; -0.5\; -0.5\;]$.\end{example}
\begin{example}Proakis-C\cite{prokc}. The 5-tap ISI channel, $\vec{h}\!=\![\;0.227\;\; 0.46\;\; 0.688\;\; 0.46\;\; 0.227\;]$.\end{example}

\begin{figure}
\begin{center}
\hspace*{-0mm}
\scalebox{0.52}{\includegraphics{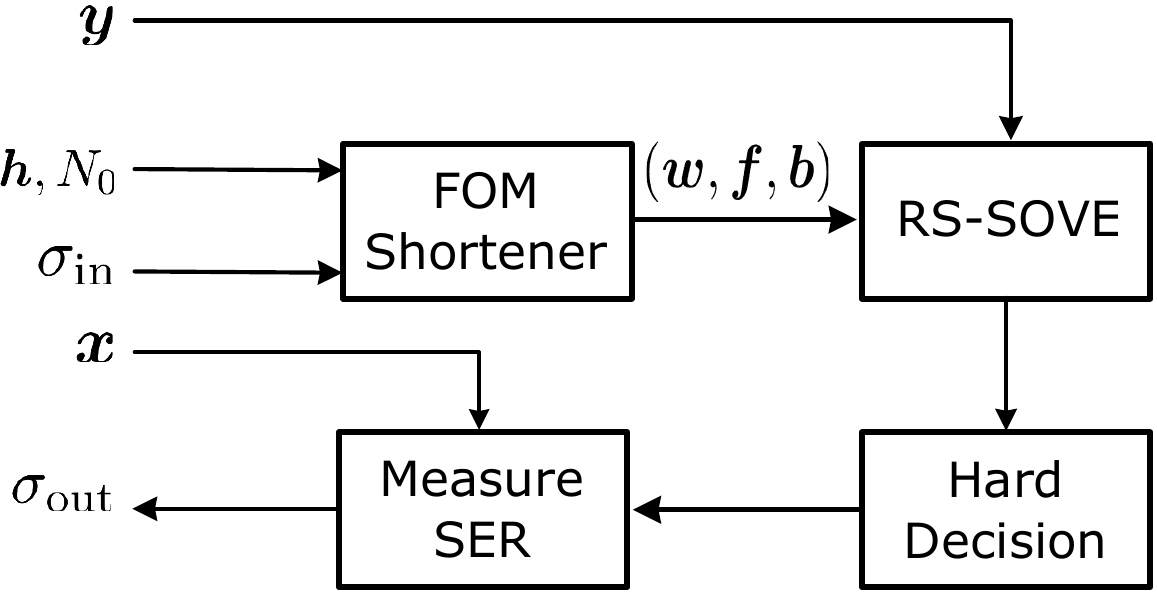}}
\vspace*{-4mm}
\caption{\label{fig4} The diagram of evaluating the optimal $\sigma$. For each input $\sigma_{\mathrm{in}}$, we calculate the optimal ($\vec{w}$, $\vec{f}$, $\vec{b}$) based on Theorem 1, and run the FOM shortener with RS-SOVE. The output $\sigma\!=\!1\!-\!P_{\mathrm{e}}$ is measured based on the hard decisions output from the RS-SOVE equalizer.}
\end{center}
\vspace*{-12mm}
\end{figure}

From Theorem 1, an input $\sigma$ determines the optimal channel shortening parameters ($\vec{w}$, $\vec{f}$, $\vec{b}$), which in turn affects the quality of the decision-feedback $\hat{\vec{x}}$ in the RS-SOVE. Therefore, there is a mismatch between the designed $\sigma$, and the practical $\sigma$ measured by the outputs of the RS-SOVE generated by such an designed $\sigma$. The output $\sigma_{\mathrm{out}}$ is measured with (\ref{sigma}), under each input $\sigma_{\mathrm{in}}$ which is utilized to generate the optimal parameters of the FOM  shortener. Under both channels, $\sigma_{\mathrm{in}}$ is increased from 0 to 1. The curves are shown in Fig. {\ref{fig5} and Fig. {\ref{fig6}, where we have two interesting observations. The first observation is that, with the FOM shortener, the RS-SOVE can only benefit from the hard decisions when the quality of the feedback is above a certain threshold, otherwise, setting $\sigma\!=\!0$, i.e., utilizing no feedback in the RS-SOVE (such as the UBM shortener) is close to optimal (also with optimized ($\vec{w}$, $\vec{f}$, $\vec{b}$) designed for $\sigma\!=\!0$). The second observation is that, when the RS-SOVE can benefit from the feedback, setting $\sigma\!=\!1$ is close to optimal, which is aligned with the analysis leading to (\ref{sigma2}).

With these observations, the design of the FOM shortener only needs to consider either $\sigma\!=\!0$ or 1. With $\sigma\!=\!0$, the UBM shortener is a more general model and has better performance than the FOM shortener. In addition, the optimization with UBM shortener is concave \cite{FDG12}. Hence, when designing the optimal channel shortener, it is sufficient to consider either the UBM shortener (\ref{UB2}), or the FOM shortener (\ref{md3}) with ($\vec{w}$, $\vec{f}$, $\vec{b}$) designed for $\sigma\!=\!1$. The remaining issue is the criterion for choosing between these two shorteners. Such a criterion is difficult to find theoretically, but as we show later through numerical results, it can be designed based on the code-rate of the considered SC systems. At medium and high code-rates, the proposed FOM channel shortener is superior to the UBM shortener.

With the HOM, FOM and UBM channel shorteners introduced in Sec. II-A, Sec. III-A, and Sec. III-B, respectively, next we analyze the mutual information (MI) characteristics. We show that the FOM shortener is superior to the HOM shortener in general, and better than the UBM shortener when the feedback $\hat{\vec{x}}$ are fairly good.

\begin{figure}
\vspace*{-8mm}
\hspace*{-3mm}
\begin{center}
\scalebox{0.42}{\includegraphics{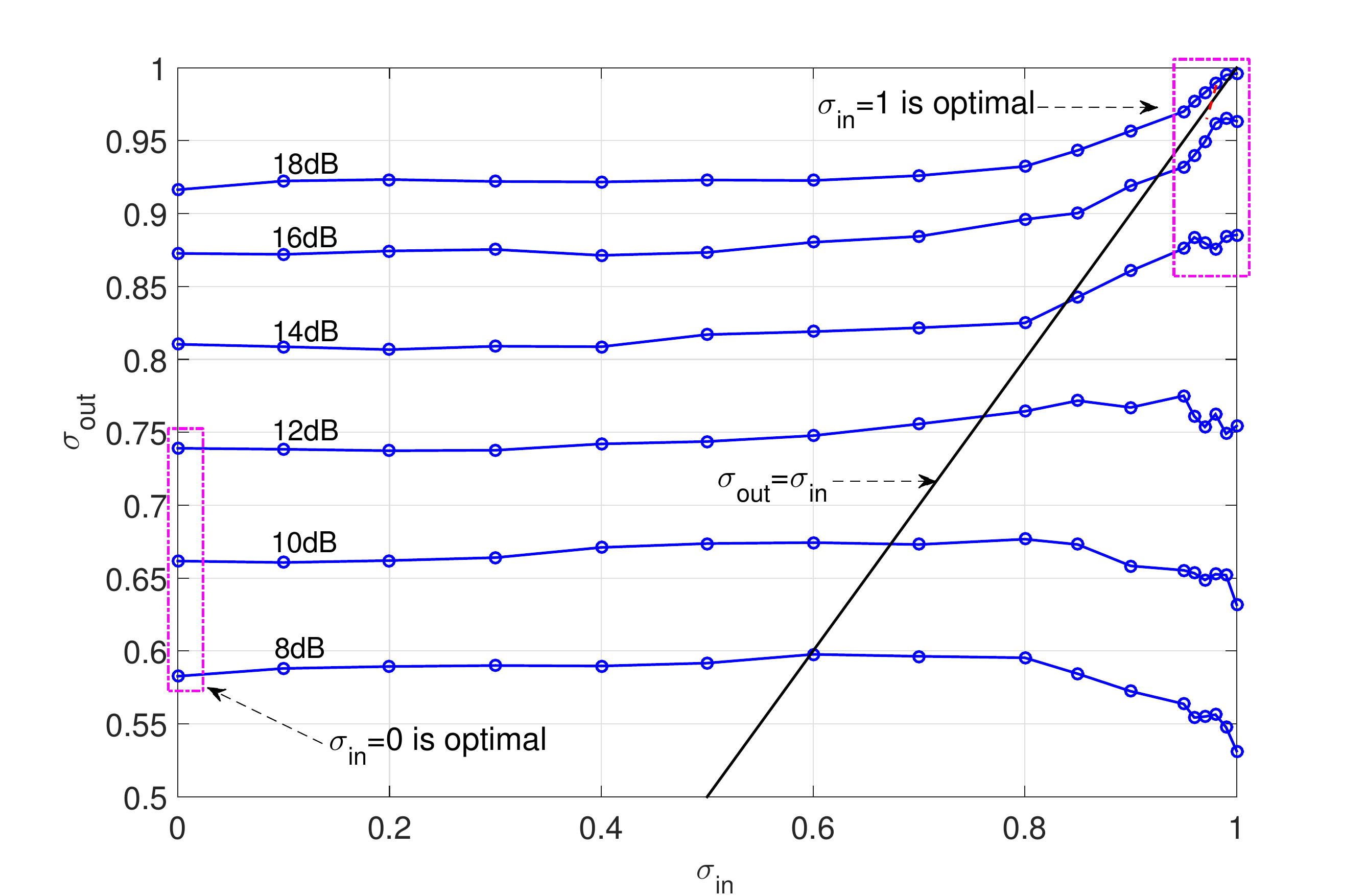}}
\vspace*{-7mm}
\caption{\label{fig5}The curves with optimal $\sigma$ investigation under EPR-4 channel and with 8PSK modulation. The $\sigma_{\mathrm{out}}$ is measured according to (\ref{sigma2}). }
\end{center}
\vspace*{-10mm}
\end{figure}

\begin{figure}
\vspace*{-6mm}
\hspace*{-3mm}
\begin{center}
\scalebox{0.42}{\includegraphics{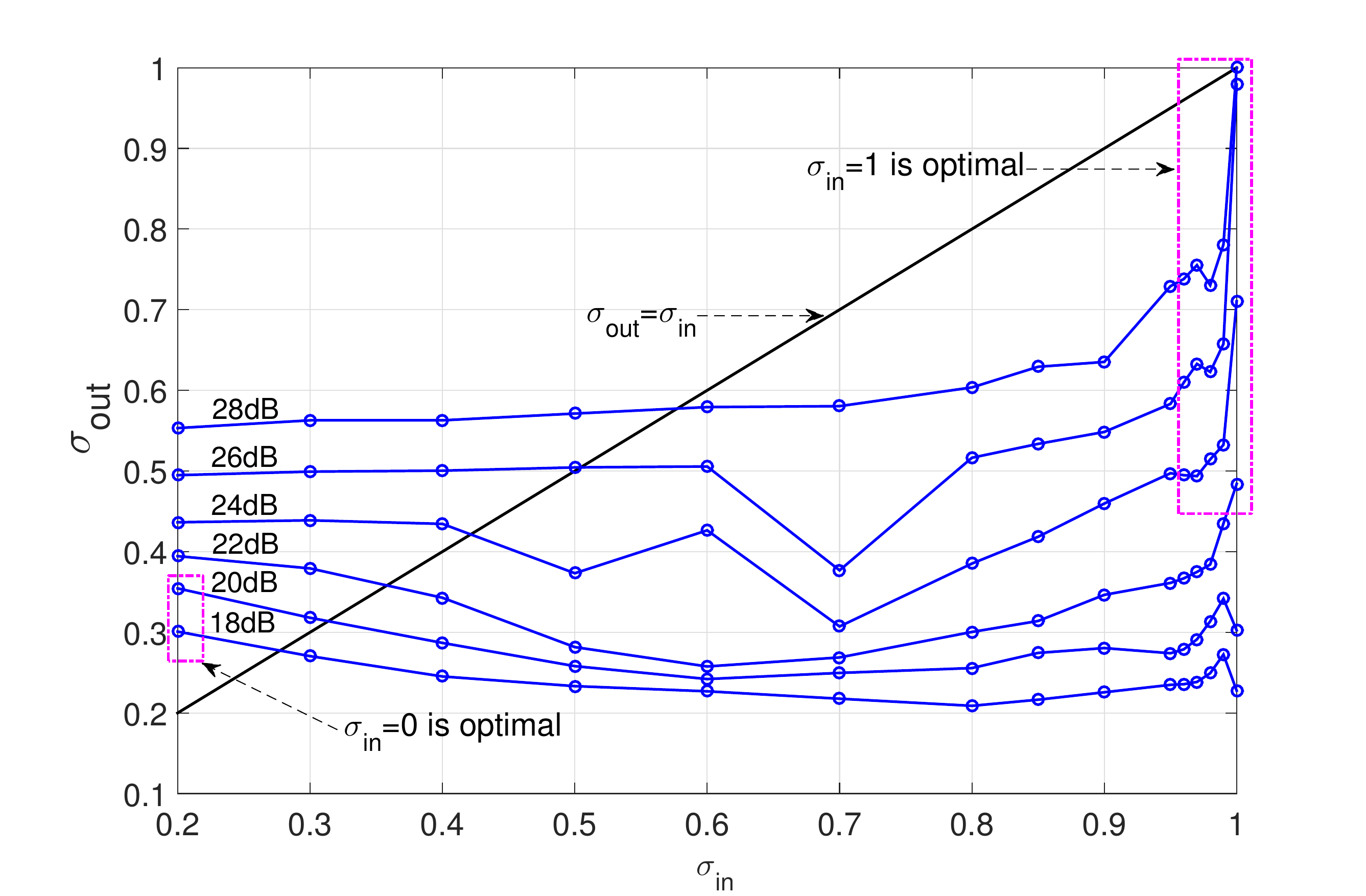}}
\vspace*{-7mm}
\caption{\label{fig6}Repeat the test in Fig. \ref{fig5} under Proakis-C channel and with 16-quadrature-amplitude-modulation (16QAM) modulation.}
\end{center}
\vspace*{-10mm}
\end{figure}

\section{Theoretical Information Rates of the Channel Shorteners}
For simplicity, we denote the optimal $I_{\mathrm{LB}}$ of the FOM and UBM shorteners as $I_{\mathrm{FOM}}$ and $I_{\mathrm{UBM}}$, calculated in (\ref{Ilb}) and (\ref{Ilb2}), respectively. Further, we denote $I_{\mathrm{FOM}}$ computed with $\sigma\!=\!0$ and 1 as $I_{\mathrm{FOM}}^0$ and $I_{\mathrm{FOM}}^1$. Similarly, we let $I_{\mathrm{HOM}}$ denote the information rate reached by the HOM shortener. Firstly, we state the below property.
\begin{property} \label{Prop1}
Denote $H_{\mathrm{f}}(\omega)$ and $H_{\mathrm{b}}(\omega)$ as the DTFTs of $\vec{h}_{\mathrm{f}}$ and $\vec{h}_{\mathrm{b}}$ in (\ref{hf}) and (\ref{hb}), respectively. Then, it holds that
\bea \label{bd1} I_{\mathrm{HOM}}^{\mathrm{L}}\leq I_{\mathrm{HOM}}\leq I_{\mathrm{HOM}}^{\mathrm{U}},\eea
where
{\setlength\arraycolsep{2pt} \bea I_{\mathrm{HOM}}^{\mathrm{L}} &=&\frac{1}{2\pi}\int_{-\pi}^{\pi}\log\left(1+\frac{|H_{\mathrm{f}}(\omega)|^2}{1+|H_{\mathrm{b}}(\omega)|^2}\right)\mathrm{d}\omega, \\
\label{prop1e1} I_{\mathrm{HOM}}^{\mathrm{U}}&=&\frac{1}{2\pi}\int_{-\pi}^{\pi}\log\left(1+|H_{\mathrm{f}}(\omega)|^2\right)\mathrm{d}\omega. \eea}
\end{property}

\begin{proof}
As $\vec{h}\!=\!\vec{h}_{\mathrm{f}}\!+\!\vec{h}_{\mathrm{b}}$, the received signal model (\ref{model3}) can be rewritten as
 \bea \tilde{\vec{y}}\ 
 =\vec{h}_{\mathrm{f}}\star\vec{x}+\vec{h}_{\mathrm{b}}\star\vec{x}+\tilde{\vec{n}}. \notag \eea
The lower bound of $I_{\mathrm{HOM}}$ is achieved when feedback $\hat{\vec{x}}$ acts as noise, while the upper bound is achieved when $\hat{\vec{x}}$ is perfect. Therefore, the inequality (\ref{bd1}) holds.
\end{proof}

More discussions about the properties of truncated channel response ${H}_{\mathrm{f}}(\omega)$ can be found in, e.g., \cite{SA98}. Here we mention the fact that, the upper bound $I_{\mathrm{HOM}}^{\mathrm{U}}$ can be higher than Shannon capacity $\mathcal{C}$, due to the perfect feedback. Secondly, we state Property 2.
\begin{property} \label{Prop2}
The below inequalities hold,
\bea \label{bd3} I_{\mathrm{HOM}}\leq I_{\mathrm{FOM}}^0\leq I_{\mathrm{UBM}}\leq \mathcal{C}.\eea
\end{property}
\begin{proof}
See Appendix C. 
\end{proof}
From Property 2, when there is no feedback, the FOM shortener is lower-bounded by the HOM shortener and upper-bounded by the UBM shortener. Further, all of them are bounded by $\mathcal{C}$. However, when $\sigma\!>\!0$, the rate of the FOM shortener can be higher than, both the UBM shortener and $\mathcal{C}$, due to the presence of feedback.

Lastly, we build the relationship between $ I_{\mathrm{HOM}}^{\mathrm{U}}$ and $I_{\mathrm{FOM}}^1$, which is stated in Property 3. Although with $\sigma\!=\!1$, the feedback $\hat{\vec{x}}$ is perfect, the symbol detection is still utilizing the received samples $\tilde{\vec{y}}$, which are not perfect. Therefore, such a comparison is meaningful and shows that, when there are no errors in $\hat{\vec{x}}$, the FOM shortener is superior to the HOM shortener.

\begin{property} \label{Prop3}
The below inequality holds,
\bea \label{bd2} I_{\mathrm{HOM}}^{\mathrm{U}} \leq I_{\mathrm{FOM}}^1. \eea
\end{property}
\begin{proof}
By setting $(\vec{w}, \vec{f}, \vec{b} )\!=\!(\vec{w}_{\mathrm{hom}}, \vec{h}_{\mathrm{f}}, \vec{h}_{\mathrm{b}})$, the FOM shortener is identical to the HOM shortener. And with $\sigma\!=\!1$, $I_{\mathrm{LB}}$ in this case equals $I_{\mathrm{HOM}}^{\mathrm{U}}$. As $I_{\mathrm{FOM}}^1$ maximizes $I_{\mathrm{LB}}$, (\ref{bd2}) holds.
\end{proof}

We summarize the above discussions in the below theorem.
\begin{theorem} \label{thm2}
The below equalities of theoretical information rates hold with $\sigma\!=\!0$,
\bea \label{thm2e1} I_{\mathrm{HOM}}^{\mathrm{L}}\leq I_{\mathrm{HOM}}\leq I_{\mathrm{FOM}}^0\leq I_{\mathrm{UBM}} \leq \mathcal{C} ,\eea
while with $\sigma\!=\!1$, the below inequalities hold,
  \bea \label{thm2eq2} I_{\mathrm{HOM}}\leq I_{\mathrm{HOM}}^{\mathrm{U}}\leq I_{\mathrm{FOM}}^1.  \eea
\end{theorem}
\begin{proof}
Combing Properties 1-3 yields Theorem 2.
\end{proof}
Theorem 2 shows that, when the quality of feedback is poor, i.e., $\sigma\!=\!0$, the UBM shortener has best performance compared to both the FOM and HOM shorteners, while when the feedback is perfect, the FOM shortener outperforms both the HOM and UBM shorteners. Moreover, as we showed earlier, the optimal $\sigma$ for the FOM shortener is either 0 or 1, hence, one can design a system that switches between the UBM shortener, and the FOM shortener designed for $\sigma\!=\!1$, to achieve the best performance under all cases.

Note that, with $\vec{f}\!=\!\vec{h}_{\mathrm{f}}$ and the optimal $\vec{w}_{\mathrm{opt}}, \vec{b}_{\mathrm{opt}}$ calculated in (\ref{optw}) and (\ref{optb}), $I_{\mathrm{LB}}$ equals
\bea \label{lb} I_{\mathrm{LB}}=1+\frac{1}{2\pi}\int_{-\pi}^{\pi}\Big(\log\left(1+|H_{\mathrm{f}}(\omega)|^2\right) +M(\omega)\left(1+|H_f(\omega)|^2\right) \Big)\mathrm{d}\omega-\vec{\varepsilon}_1\rmh\vec{\varepsilon}_2^{-1}\vec{\varepsilon}_1.  \eea
By definition, $I_{\mathrm{LB}}$ in (\ref{lb}) is no less than $I_{\mathrm{LB}}$ computed with $(\vec{w}_{\mathrm{hom}}, \vec{h}_{\mathrm{f}}, \vec{h}_{\mathrm{b}})$, which equals $I_{\mathrm{HOM}}^{\mathrm{U}}$. From (\ref{prop1e1}) and (\ref{lb}), we have an interesting corollary below that shows the relation between $\vec{h}$ and $\vec{h}_{\mathrm{f}}$ for any ISI channels, and reveals the fact that, with the same target response $\vec{f}$ but optimized $\vec{w}$, $\vec{b}$, the FOM shortener outperforms the HOM shortener.

\begin{corollary}
For any ISI channel $\vec{h}$ and the target response $\vec{h}_{\mathrm{f}}$ defined in (\ref{hf}), the inequality 
{\setlength\arraycolsep{2pt} \bea \label{deltaI} \vec{\varepsilon}_1\rmh\vec{\varepsilon}_2^{-1}\vec{\varepsilon}_1-\frac{1}{2\pi}\int_{-\pi}^{\pi}M(\omega)\left(1+|H_{\mathrm{f}}(\omega)|^2\right)\mathrm{d}\omega \leq 1,\notag \eea}
\hspace{-1.4mm}holds, where $\vec{\varepsilon}_1$, $\vec{\varepsilon}_1$ are defined in (\ref{vareps1}) and (\ref{vareps2}) with $F(\omega)\!=\!H_{\mathrm{f}}(\omega)$, and $M(\omega)$, $\tilde{M}(\omega)$ are defined in (\ref{mw}) and (\ref{mtw}) with $\sigma\!=\!1$, respectively.
\end{corollary}

\section{Empirical Results}
In this section, we provide empirical results to show the information rates and detection performance of the proposed FOM channel shortener with the RS-SOVE, and compare it to the UBM and HOM shorteners. Throughout all tests, without explicitly pointing out, we assume that the memory length $\nu\!=\!1$ after channel shortening to achieve a low-complexity receiver design.

For each transmit symbol vector $\vec{x}$, with different channel shorteners, the bit LLRs $L(x_{k,n})$ are calculated in (\ref{llr}) based on different branch metric computations as in (\ref{metric1}), (\ref{metric2}), and (\ref{metric3}), respectively. As the transmit bits $x_{k,n}$ are independent, the logarithm of the conditional probability of each symbol $x_k\rq{}\!\in\!\mathcal{X}$ for a given transmit symbol $x_k$, i.e., $p(x_k\rq{}|x_k)$, can be computed as
{\setlength\arraycolsep{2pt}\bea \label{pxn}\log p(x_k\rq{}|x_k)  &=&\sum\limits_{n=0}^{\log_2 |\mathcal{X}|-1}\log p(x_{k,n}\rq{}|x_{k,n}) \notag \\
&=&\sum_{n=0}^{\log_2 |\mathcal{X}|-1}\Big(\frac{(1+x_{k,n}\rq)L(x_{k,n})}{2}-\log\big(1+\exp\left(L(x_{k,n})\right)\big)\Big).\qquad \notag \eea}
\hspace*{-1.4mm}Then, the measured MI is calculated as
 \bea I(\vec{y};\vec{x})=\log |\mathcal{X}|-\mathbb{E}_{x_k,x_k\rq{}\in\mathcal{X}}\big[\log p(x_k\rq{}|x_k)\big].  \notag \eea

\subsection{The Impact of Decision-Delay $D$ in RS-SOVE}
First, we evaluate the normalized MI measured for the EPR-4 channel, and investigate the impact of decision-delay $D$ for different modulation schemes in RS-SOVE. The HOM shortener is tested with $D$ set to $L\!-\!1$, $L\!+\!2$ and $L\!+\!20$, respectively. As can be seen in Fig. \ref{fig7}, with 16QAM modulation, enlarging $D$ from $L\!-\!1$ to $L\!+\!2$ has around an SNR gain of 0.4 dB in terms of the normalized MI. However, further increasing $D$ up to $L\!+\!20$ only has marginal SNR gain. Since a larger delay increases process latency in the RS-SOVE, in the remaining tests we set $D\!=\!L\!+\!2$ in the RS-SOVE for all channel shorteners.

\begin{figure}
\hspace{24mm}
\vspace*{-4mm}
\scalebox{.42}{\includegraphics{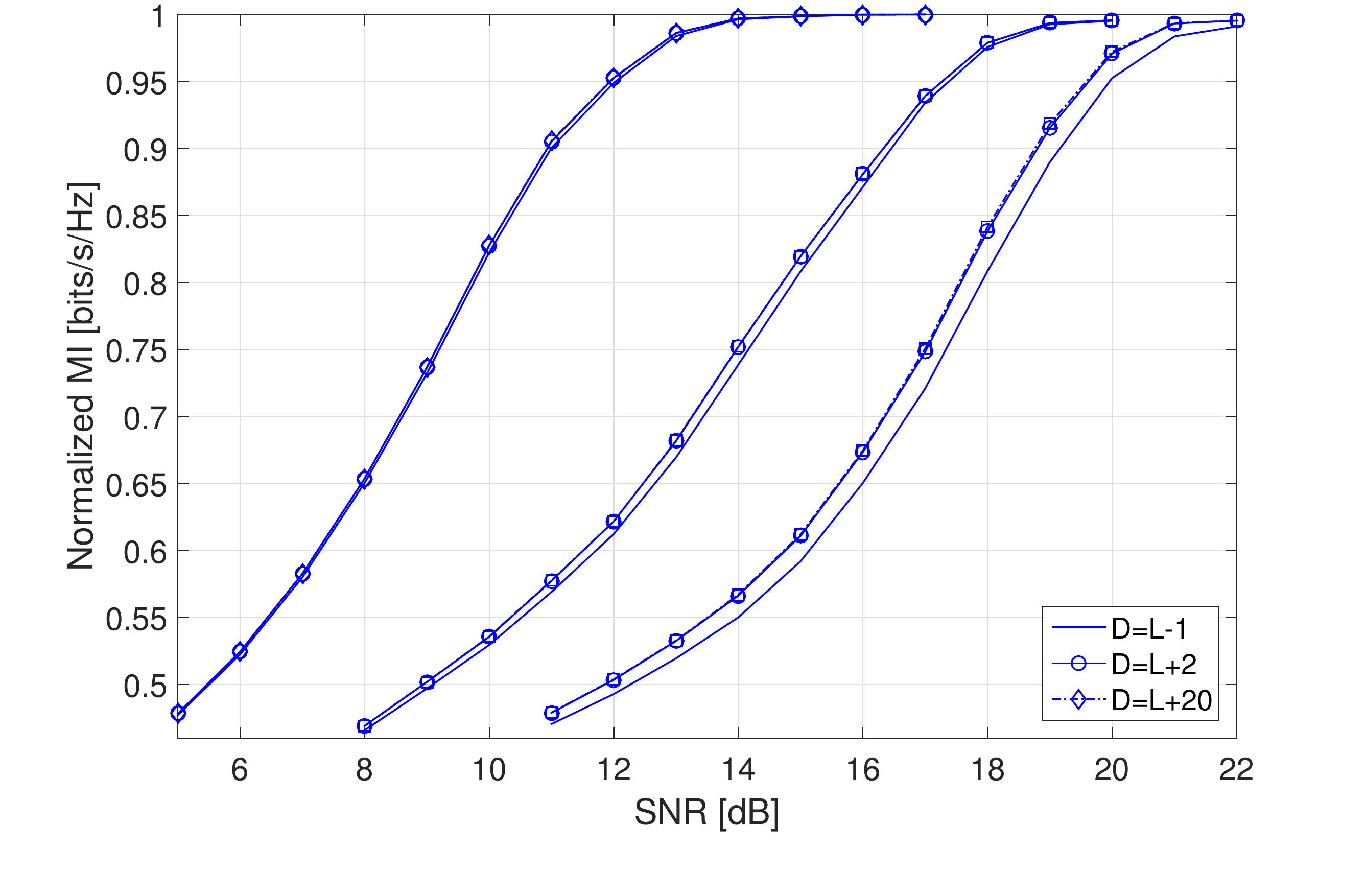}}
\vspace*{-8mm}
\caption{\label{fig7}Performance evaluation of the HOM shortener with different delays and modulation schemes. From left to right, the modulation schemes are quadrature-phase-shift keying (QPSK), 8PSK and 16QAM. The normalized MI is measured with the output from the RS-SOVE.}
\vspace*{-7mm}
\end{figure}

\subsection{Theoretical Information Rates}
Next, we simulated the theoretical information rates that have been discussed in Sec. IV under EPR-4 and Proakis-C channels. In comparison, we also add the rates of $I_{\mathrm{FOM}}$ with $\sigma\!=\!1\!-\!\delta_{\mathrm{mse}}/2$. The rates of LMMSE detection and Shannon capacity $\mathcal{C}$ are also presented. As can be seen, both in Fig. \ref{fig8} and Fig. \ref{fig9}, with $\sigma\!=\!0$ the information rates of the FOM shortener ($I_{\mathrm{FOM}}^0$) and the UBM shortener ($I_{\mathrm{UBM}}$) are quite close. In the low SNR regime, the UBM shortener is superior, while in the high SNR regime, the information rates of the FOM shortener with $\sigma\!=\!1$ ($I_{\mathrm{FOM}}^1$) are the best. Under EPR-4 channel, $I_{\mathrm{FOM}}^1$ and $I_{\mathrm{HOM}}^{\mathrm{U}}$ asymptotically align with $\mathcal{C}$, while under Proakis-C channel, both $I_{\mathrm{FOM}}^1$ and $I_{\mathrm{HOM}}^{\mathrm{U}}$ are higher than $\mathcal{C}$. Moreover, when SNR increases, $I_{\mathrm{HOM}}^{\mathrm{U}}$ asymptotically approaches the rate $I_{\mathrm{FOM}}^1$.

As the differences of the information rates between the FOM shortener with $\sigma\!=\!0$ and the UBM shortener cannot be seen clearly in Fig. {\ref{fig8} and Fig. {\ref{fig9}, in Fig. {\ref{fig10} we normalize $I_{\mathrm{FOM}}^0$ with $I_{\mathrm{UBM}}$. In addition, we also add the results of another type of ISI channel stated in Example 3. The results show that the FOM shortener is slightly inferior to the UBM shortener, which is aligned with Property 3.
\begin{example}A 5-tap IID complex Gaussian channel with unit energy per realization..\end{example}

\vspace{-6mm}

\subsection{Measured MI}
In order to verify the practical performance, we measure the MI achieved by the three shorteners with different modulation schemes and under Proakis-C channels. As can be seen from Fig. \ref{fig11}, the measured MI results are aligned with the theoretical analysis illustrated in Fig. \ref{fig9}. The UBM shortener outperforms both the FOM and HOM shorteners in the low SNR regime. But when SNR increases, the FOM shortener becomes the best. The HOM shortener is in general inferior to the FOM shortener, and in the high SNR regime, the HOM shortener performs close to the FOM shortener. We also add the information rates of the FOM shortener with both $\sigma\!=\!1\!-\!\delta_{\mathrm{mse}}/2$ and $\sigma\!=\!1\!-\!\delta_{\mathrm{mse}}/4$, which are inferior to the rates of the FOM shortene with $\sigma\!=\!1$ in the high SNR regime. 

Most interestingly, the cross points of the FOM shortener with $\sigma\!=\!1$ and the UBM shortener are below 1/2 in terms of the normalized MI, which indicates that, the switching criterion of the FOM and UBM shorteners can be based on the output MI of the RS-SOVE, or equivalently, the input MI to the outer-decoder. As for error-correcting codes, the input MI of the LLRs sent to the decoders shall be no less than the code-rate for successfully decoding. Therefore, we can use the code-rate as the criterion. If the code-rate is higher than 1/2, the proposed FOM shortener will provide better performance, otherwise we switch to the UBM shortener. This is also due to the fact that, the FOM shortener is superior to the UBM shortener only when the feedback quality is fairly good.

\begin{figure}
\hspace{-3mm}
\vspace{-4mm}
\begin{center}
\scalebox{0.42}{\includegraphics{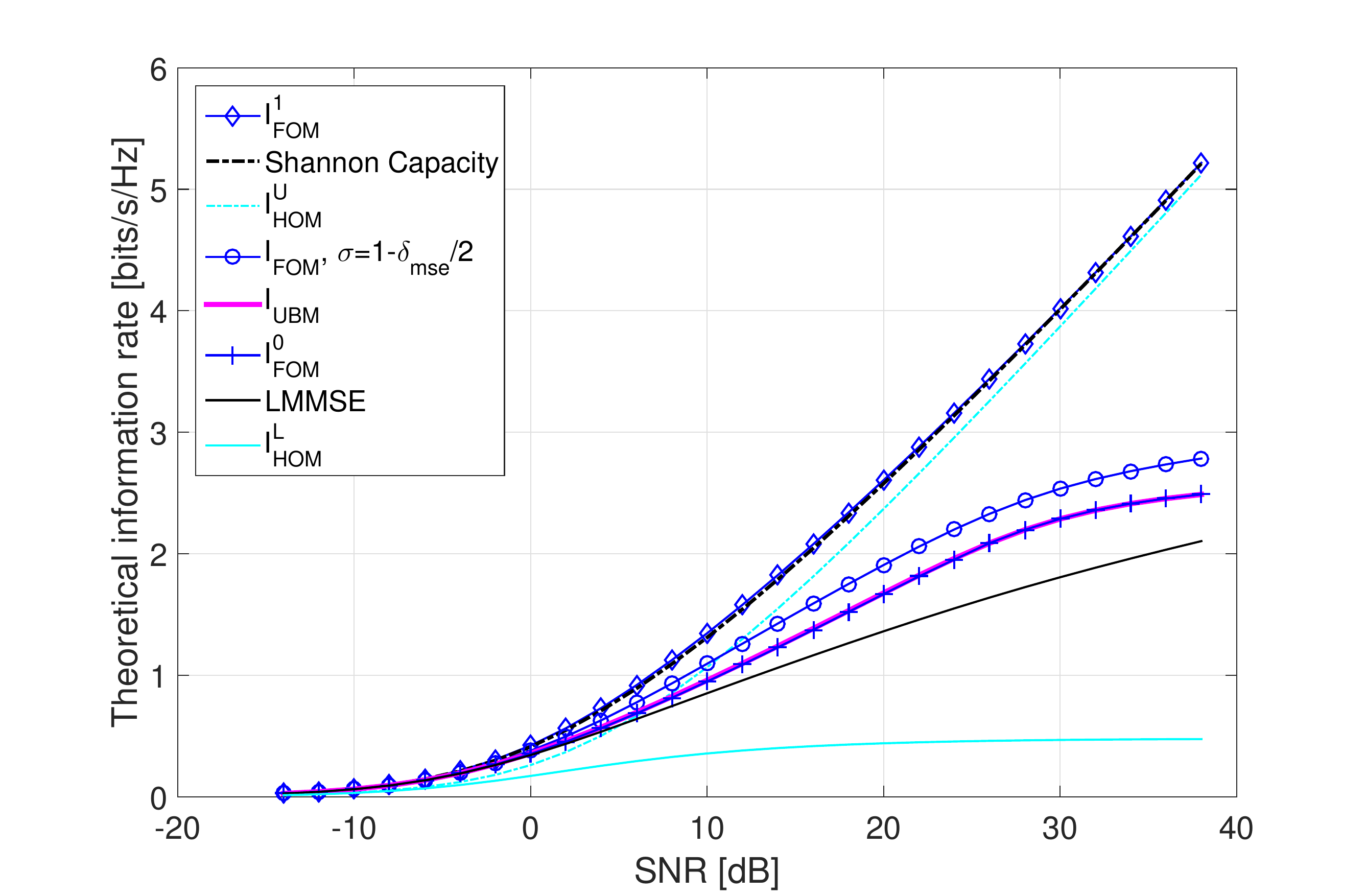}}
\vspace{-6mm}
\caption{\label{fig8}Theoretical information rates under EPR-4 channel (the legend is ordered from the top curve to the bottom curve).}
\end{center}
\vspace{-10mm}
\end{figure}

\begin{figure}
\hspace{-3mm}
\vspace{-3mm}
\begin{center}
\scalebox{0.42}{\includegraphics{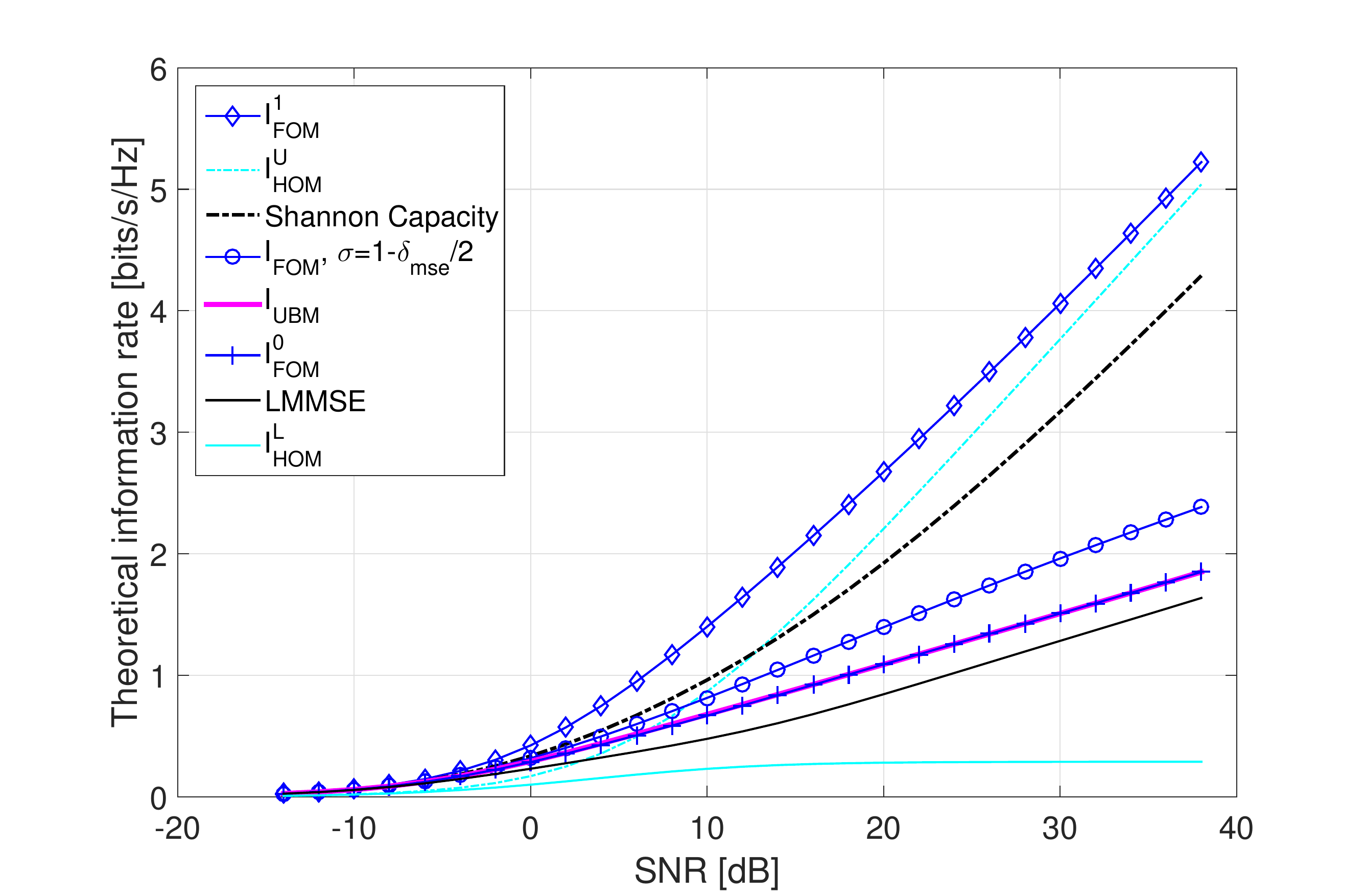}}
\vspace*{-6mm}
\caption{\label{fig9}Repeat the test in Fig. {\ref{fig8} under Proakis-C channel.}}
\end{center}
\vspace{-9mm}
\end{figure}

\begin{figure}
\hspace*{-3mm}
\vspace*{-4mm}
\begin{center}
\scalebox{0.42}{\includegraphics{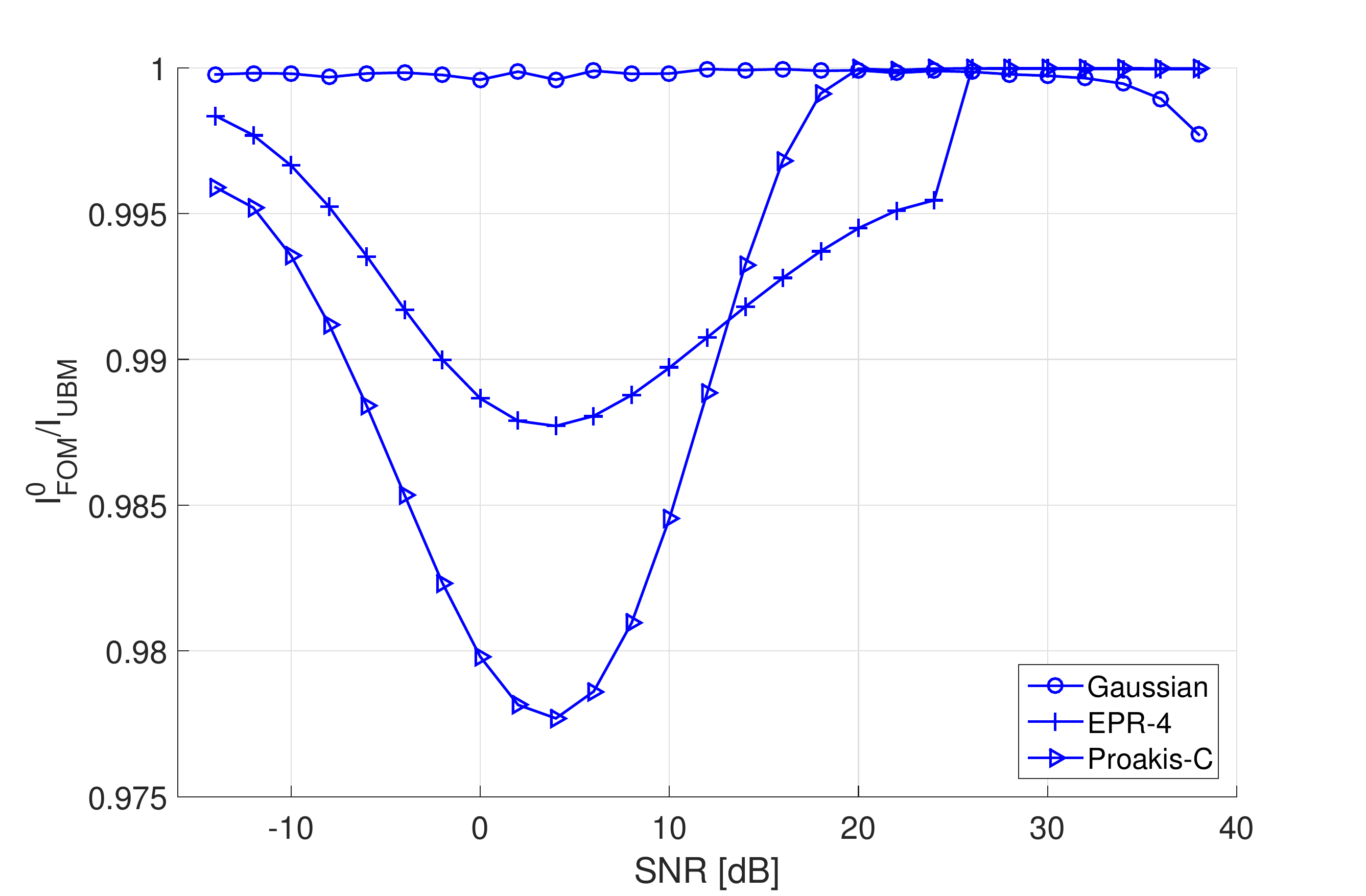}}
\vspace*{-6mm}
\caption{\label{fig10}The normalized information rates of the FOM shortener with $\sigma\!=\!0$ over the rates of the UBM shortener under EPR-4, Proakis-C, and 5-tap IID Guassian channels.}
\end{center}
\vspace*{-12mm}
\end{figure}

\subsection{Parameter Optimization of the FOM Channel Shortener }
Next, we evaluate the parameter optimization of the FOM channel shortener. As stated in Theorem 1, the optimal prefilters $\vec{w}$ and $\vec{b}$ are in closed forms, while the optimal $\vec{f}$ has to be found through an optimization process. 

In Fig. \ref{fig14}, we plot the convergence speed under EPR-4 and Proakis-C channels at different SNR points. We test with $\sigma\!=\!1/2$ and $\sigma\!=\!1$, respectively. As can been seen, the optimization converges very fast in a few number of iterations.

\subsection{Performance Evaluation with Turbo Codes}
At last, we evaluate the BER performance with turbo codes specified in LTE standard \cite{Tdec}. we set the number of information bits $K\!=\!1064$ for all tests, and evaluate different code-rates and modulation schemes. In Fig. \ref{fig15}, we show the BER results under EPR-4 channel and with 8PSK modulation. As expected, the UBM shortener performs the best at code-rates 1/3 and 1/2. At higher code-rates 2/3 and 3/4, the UBM shortener becomes inferior to the FOM shortener. In all cases, the FOM shortener is superior to the HOM shortener. 

In Fig. \ref{fig16}, the BER results under Proakis-C channel and with 16QAM modulation are presented. In this case, the UBM shortener outperforms the other two channel shorteners at code-rate 1/3 only. At higher code-rates, the UBM shortener tends to perform poorly. However, the proposed FOM shortener is still around 1-2 dB better in terms of SNR than the HOM shortener at all code-rates. These results are also aligned with Fig. \ref{fig11}, where we show that, the UBM shortener outperforms the FOM shortener only when the normalized MI below around 1/2, while with higher MI, the UBM shortener is inferior.

\begin{figure}[t]
\hspace{-15mm}
\vspace{-6mm}
\scalebox{0.54}{\includegraphics{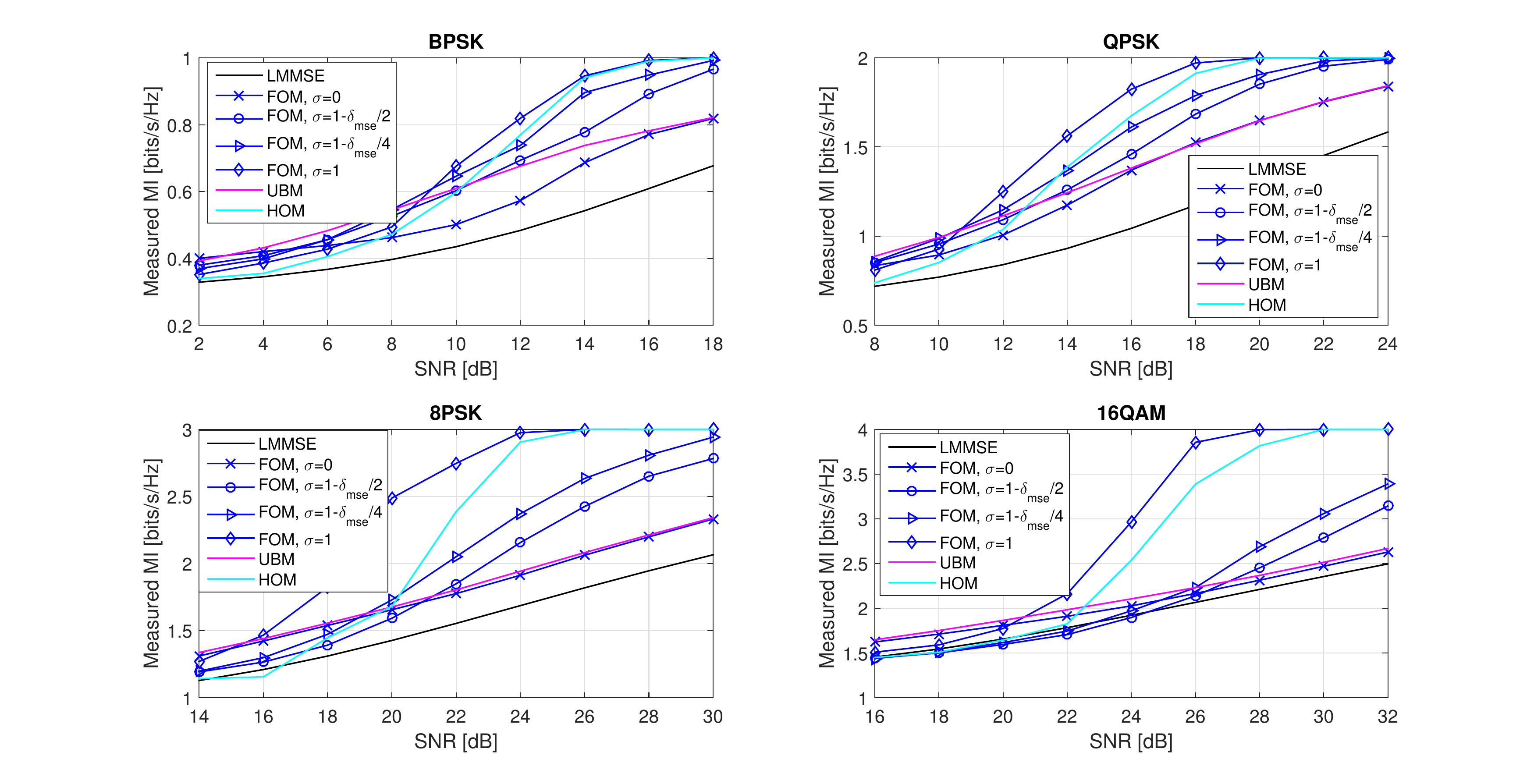}}
\vspace{-11mm}
\caption{\label{fig11}Measured MI under Proakis-C channel and with different modulation schemes. The UBM shortener provides the best performances when the normalized MI is lower than around 1/2, while the FOM shortener with $\sigma\!=\!1$ is the best for normalized MI higher than 1/2. The conventional HOM shortener is in general interior to the FOM shortener, except that in the high SNR regime it approaches the rates of the FOM. Moreover, the cross-points between the FOM and UBM shorteners are around 1/2 in terms of the normalized MI.}
\vspace{-8mm}
\end{figure}

\begin{figure}
\vspace*{-0mm}
\begin{center}
\hspace*{-2mm}
\scalebox{.375}{\includegraphics{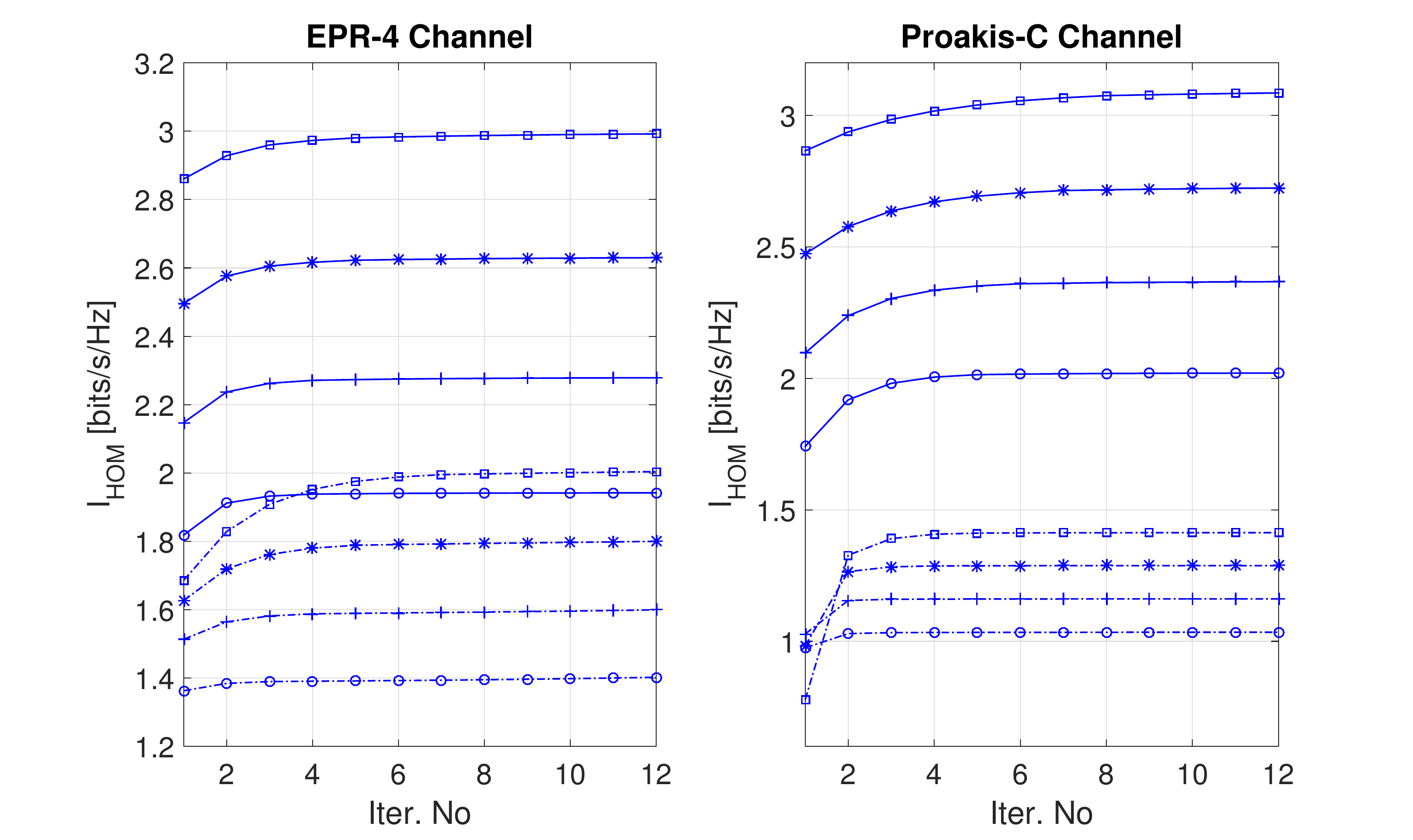}}
\vspace*{-6mm}
\caption{\label{fig14}The convergence speed of the FOM shortener. The dashed lines are with $\sigma\!=\!1/2$ while the solid lines are with $\sigma\!=\!1$. In both cases and from bottom to up, the SNR equals 10dB, 12dB, 14dB, and 16dB, respectively. With larger $\sigma$, the optimization need more steps to converge. However, as can be seen in both figures, the optimization process converges in 4-8 iterations.}
\vspace*{-8mm}
\end{center}
\end{figure}
\begin{figure}
\vspace*{-1mm}
\hspace*{22mm}
\scalebox{.42}{\includegraphics{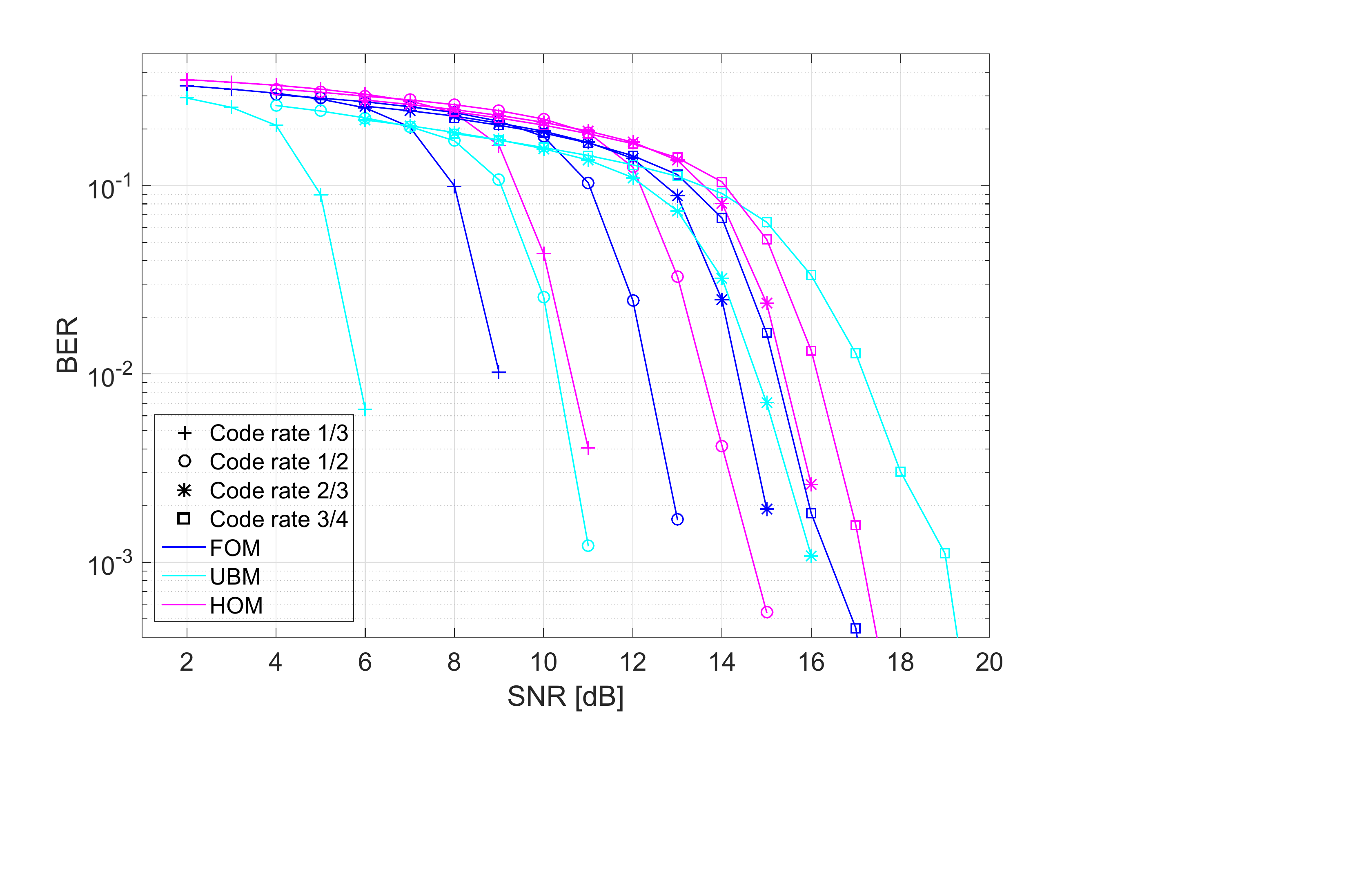}}
\vspace*{-33mm}
\caption{\label{fig15}The coded BER evaluation with turbo codes under EPR-4 channel. At higher code-rate 2/3 and 3/4, the FOM shortener is superior to the UBM shortener, while at all code-rates, the FOM shortener is better than the HOM shortener. }
\vspace*{-4mm}
\end{figure}
\begin{figure}
\hspace{22mm}
\vspace{-0mm}
\scalebox{.42}{\includegraphics{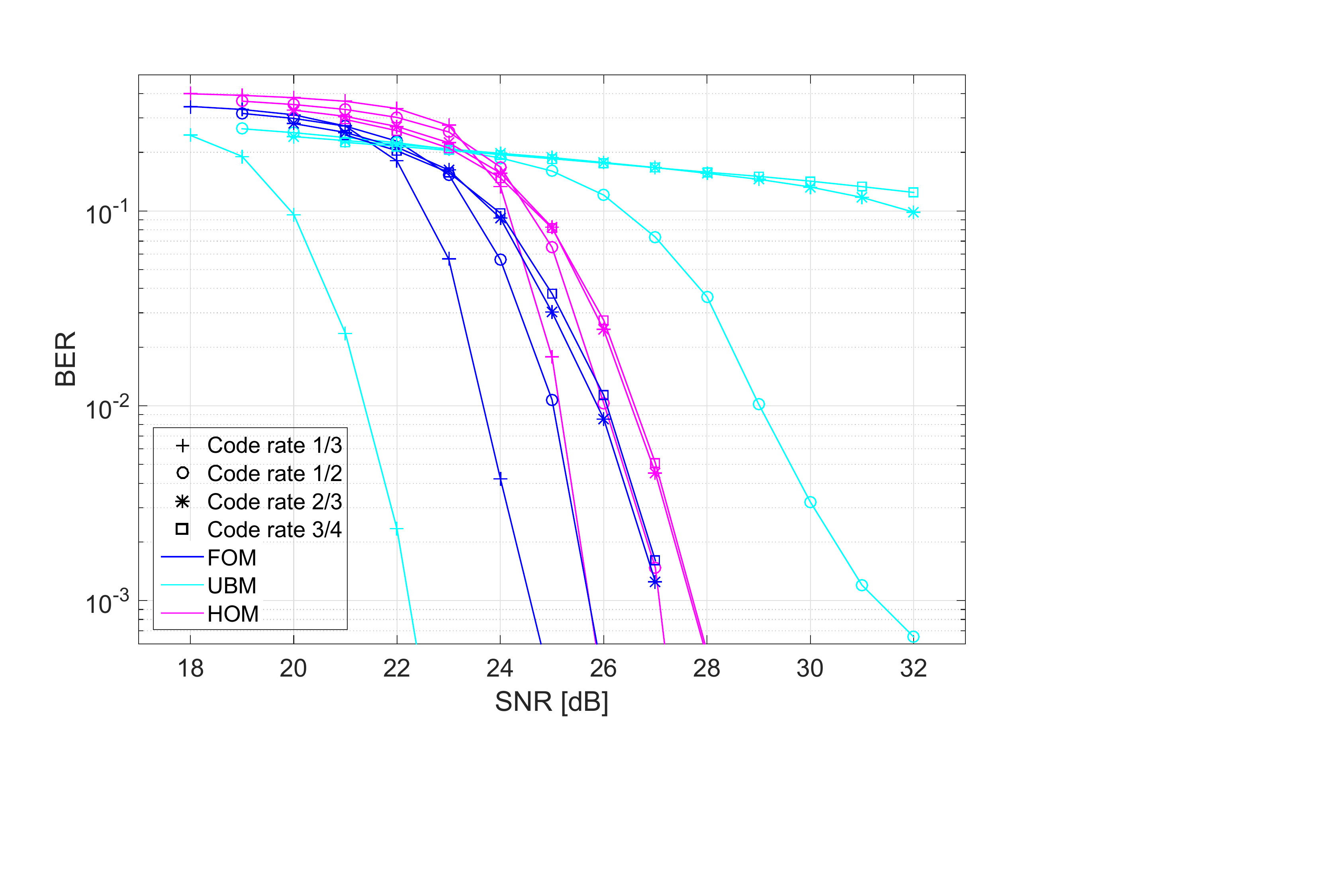}}
\vspace*{-33mm}
\caption{\label{fig16}The coded BER evaluation with turbo codes under Proakis-C channel. Truncating the channel tails in the UBM shortener renders significantly performance losses at code-rates higher than 1/2. The FOM shortener is better than the HOM shortener for all code-rates.}
\vspace*{-6mm}
\end{figure}

\section{Summary}
In this paper, we consider the mutual information lower bound (MILB) based channel shortener design that used in conjunction with the reduced-state soft-output Viterbi equalizer (RS-SOVE), namely, the FOM shortener. We show that the FOM channel shortener cooperating with the RS-SOVE has major gains over the Ungerboeck detection model based channel shortener, namely, the UBM shortener, at medium and high code-rates. Due to the lack of probabilistic meaning, the UBM shortener truncates the channel tails and utilizes no decision-feedback detection. Both the FOM and UBM shorteners significantly outperform the conventional homomorphic filtering based channel shortener, namely, the HOM shortener. We also analyze the theoretical information rates of the proposed FOM channel shortener in relation to the Shannon capacity and the previous channel shortener designs. In addition, we extend the RS-SOVE to an arbitrary delay that can be larger than the duration of the intersymbol interference (ISI) channel, and we show that, the trellis search process is equivalent to a full forward recursion and a backward recursion with a depth that equals the delay.

\section*{Appendix A: The Proof of Theorem 1}
The DTFT of $\vec{w}$ reads
\bea \label{Ww} W(\omega)=\sum_{k=-\infty}^{\infty}w_k\exp\!\left(j k\omega \right), \notag\eea
and the differential of $I_{\mathrm{LB}}$ in (\ref{ibarm1}) with respect to $w_k$ is
{\setlength\arraycolsep{2pt}  \bea \label{wwopt} \frac{\partial I_{\mathrm{LB}}}{\partial w_k}&=&-\frac{1}{2\pi}\! \int_{-\pi}^{\pi}\!\frac{|F(\omega)|^{2}\big(N_0\!+\!|H(\omega)|^2\big)W^{\ast}(\omega)}{1+ |F(\omega)|^2}\exp\!\left(j k\omega \right)\mathrm{d}\omega\!  \notag \\
&&+ \frac{1}{\pi}\!\int_{-\pi}^{\pi}\!\!\Big(F^{\ast}(\omega)H(\omega)\!+\! \frac{\sigma |F(\omega)|^{2}H(\omega)B^{\ast}(\omega)}{1\!+\! |F(\omega)|^2}\Big)\!\exp\!\big(j k\omega \big)\mathrm{d}\omega.   \eea}
\hspace{-1.4mm}As (\ref{wwopt}) shall equal zero for all $k$, the optimal $W(\omega)$ is given in (\ref{optw}). Inserting $W_{\mathrm{opt}}(\omega)$ back into (\ref{ibarm1}) yields,
{\setlength\arraycolsep{2pt} \bea \label{ibarm1optw} I_{\mathrm{LB}}&=&1\!+\!\frac{\sigma}{\pi}\!\int_{-\pi}^{\pi}\!\!\mathcal{R}\!\big\{F^{\ast}(\omega)B(\omega)M(\omega) \big\}\mathrm{d}\omega \notag \\
&&+\frac{1}{2\pi}\! \int_{-\pi}^{\pi}\!\!\Big(\!\log\!\big(1\!+\! |F(\omega)|^2\big)+\frac{\tilde{M}(\omega)|B(\omega)F(\omega)|^2}{1\!+\!|F(\omega)|^2}\!+\!M(\omega)\big(1\!+\!|F(\omega)|^2\big) \!\Big)\mathrm{d}\omega.  \eea}
\hspace{-1.4mm}Setting $\sigma\!=\!0$, $I_{\mathrm{LB}}$ in (\ref{ibarm1optw}) equals (\ref{jw}). With $0\!<\!\sigma\!\leq\!1$, the terms related to $B(\omega)$ in (\ref{ibarm1optw}) are
 \bea \label{isif2} \mathcal{F}(B(\omega))&=&\frac{\sigma}{\pi}\!\int_{-\pi}^{\pi}\mathcal{R}\big\{F^{\ast}(\omega)B(\omega)M(\omega) \big\}\mathrm{d}\omega +\frac{1}{2\pi}\! \int_{-\pi}^{\pi}\!\frac{\tilde{M}(\omega)|B(\omega)F(\omega)|^2}{1\!+\!|F(\omega)|^2}\mathrm{d}\omega. \eea
With $\vec{\varepsilon}_1$, $\vec{\varepsilon}_2$ defined in (\ref{vareps1}) and (\ref{vareps2}), (\ref{isif2}) can be rewritten as
\bea  \label{bw1} \mathcal{F}(B(\omega))=\vec{b}\vec{\varepsilon}_2\vec{b}\rmh\!+\!2\mathcal{R}\big\{\vec{b}\vec{\varepsilon}_1 \big\}. \eea
Optimizing (\ref{bw1}) directly yields
\bea \label{bwopt} \vec{b}_{\mathrm{opt}}=-\vec{\varepsilon}_1\rmh\vec{\varepsilon}_2^{-\!1}. \notag \eea
Then the optimal $B_{\mathrm{opt}}(\omega)$ is given in (\ref{optb}). Inserting $B_{\mathrm{opt}}(\omega)$ back into (\ref{ibarm1optw}), $I_{\mathrm{LB}}$ for the optimal $W_{\mathrm{opt}}(\omega)$ and $B_{\mathrm{opt}}(\omega)$, after some manipulations, is in (\ref{Ilb}).

\section*{Appendix B: Proof of Inequality ($\mathrm{a}$) in (\ref{sigma0}) }
Assume that $\vec{e}\!=\!\hat{\vec{x}}^{\mathrm{LMMSE}}-\hat{\vec{x}}$, where $\hat{\vec{x}}$ are the hard decisions corresponding to LMMSE esimates $\hat{\vec{x}}^{\mathrm{LMMSE}}$. Then,
{\setlength\arraycolsep{2pt}  \bea  
&&\mathbb{E}\left[\left(\vec{x}-\hat{\vec{x}}^{\mathrm{LMMSE}}\right)\left(\vec{x}-\hat{\vec{x}}^{\mathrm{LMMSE}}\right)^\dag\right]  \notag \\
&&=\mathbb{E}\left[\left(\vec{x}-\hat{\vec{x}}-\vec{e}\right)\left(\vec{x}-\hat{\vec{x}}-\vec{e}\right)^\dag\right]  \notag \\
&&=\left(1-P_e^{\mathrm{LMMSE}}\right)\mathbb{E}\left[\vec{e}\vec{e}^\dag\right] +P_e^{\mathrm{LMMSE}}\left(\mathbb{E}\left[\left(\vec{x}-\hat{\vec{x}}\right)\left(\vec{x}-\hat{\vec{x}}\right)^\dag\right]+\mathbb{E}\left[\vec{e}\vec{e}^\dag\right] \right) \notag \\
&& \geq P_e^{\mathrm{LMMSE}}\mathbb{E}\left[\left(\vec{x}-\hat{\vec{x}}\right)\left(\vec{x}-\hat{\vec{x}}\right)^\dag\right].  \notag \eea}
\hspace{-1.4mm}Assuming $\vec{x}$ and $\hat{\vec{x}}$ are independent for higher-order modulations, it holds that
 \bea \mathbb{E}\left[\left(\vec{x}-\hat{\vec{x}}\right)\left(\vec{x}-\hat{\vec{x}}\right)^\dag\right]=\mathbb{E}\left[\vec{x}\vec{x}^\dag\right]+\mathbb{E}\left[\hat{\vec{x}}\hat{\vec{x}}^\dag\right]=2. \notag \eea
Therefore, the below inequality holds,
\bea P_{\mathrm{e}}^{\mathrm{LMMSE}}\leq  \mathbb{E}\left[\left(\vec{x}-\hat{\vec{x}}^{\mathrm{LMMSE}}\right)\left(\vec{x}-\hat{\vec{x}}^{\mathrm{LMMSE}}\right)^\dag\right]\Big/ 2 .\notag \eea

\section*{Appendix C: Proof of Property 2 }

As the HOM shortener is a special case of the FOM shortener, by definition $I_{\mathrm{HOM}}\!\leq\!I_{\mathrm{FOM}}^0$ holds. With $\sigma\!=\!0$ and from Theorem 1, by identifying $~G(\omega)\!=\!|F(\omega)|^2~$, $I_{\mathrm{LB}}$ can be written in the same form as in (\ref{Ilb2}). As the UBM shortener maximizes (\ref{Ilb2}) under constraint that $1\!+\!G(\omega)\!\geq\! 0$ for all $\omega$, which is also true for stetting $G(\omega)\!=\!|F(\omega)|^2$,  therefore, $I_{\mathrm{FOM}}^0\!\leq\! I_{\mathrm{UBM}}$ holds. 

Next, we prove $ I_{\mathrm{UBM}} \!\leq\! \mathcal{C}$. Note that,
 \bea G(\omega)=2\mathcal{R}\left\{g_0+\sum_{k=1}^{\nu}g_k\exp\!\big(jk\omega \big)\right\}.\notag \eea
Taking the differential of $I_{\mathrm{UBM}}$ in (\ref{Ilb2}) with respect to $g_k$ and $g_k^{\ast}$results in
 \bea \int_{-\pi}^{\pi}\!\frac{\exp\!\big(jk\omega \big)}{1+G(\omega)}\mathrm{d}\omega\!=\!-\int_{-\pi}^{\pi}M(\omega)\exp\!\big(jk\omega \big)\mathrm{d}\omega,\; -\nu\leq k \leq\nu. \notag  \eea
Hence, the below equality holds with the optimal $G(\omega)$, which we denote as $G_0(\omega)$,
 \bea \label{gm} \frac{1}{1+G_0(\omega)}+M(\omega)=2\mathcal{R}\left\{\sum_{|k|>\nu}\tau_k\exp\!\big(jk\omega \big)\right\},  \eea
for some constants $\tau_k$. On the other hand, as
 \bea \label{g0} G_0(\omega)=2\mathcal{R}\left\{\hat{g}_0+\sum_{k=1}^{\nu}\hat{g}_k\exp\!\big(jk\omega \big)\right\}, \eea
for some $\hat{\vec{g}}\!=\!(\hat{g}_0,\hat{g}_1,\cdots,\hat{g}_{\nu})$, multiplying both sides in (\ref{gm}) with $\left(1+G_0(\omega)\right)$ results in
 \bea \label{gm2} 1+M(\omega)\left(1+G_0(\omega)\right)
=-2\left(1+G_0(\omega)\right)\!\mathcal{R}\left\{\sum_{k>\nu}\tau_k\exp\!\big(jk\omega \big)\right\}. \eea
 \hspace{-1.2mm}Integrating (\ref{gm2}) over $\omega$ in $[-\pi,\pi)$ and utilizing (\ref{g0}) lead to
  \bea \label{gm3}\frac{1}{2\pi} \int_{-\pi}^{\pi}M(\omega)\left(1+G_0(\omega)\right)\mathrm{d}\omega=-1. \eea
 Therefore, with $G_0(\omega)$, $ I_{\mathrm{LB}}$ in (\ref{Ilb2}) equals
 \bea I_{\mathrm{UBM}}=\frac{1}{2\pi}\int_{-\pi}^{\pi}\log(1+G_0(\omega))\mathrm{d}\omega. \notag \eea
As the logarithm function is concave, from the definition of $M(\omega)$ in (\ref{mw}) and utilizing (\ref{gm3}),
{\setlength\arraycolsep{2pt}\bea I_{\mathrm{UBM}}-\mathcal{C} &=&\frac{1}{2\pi}\int_{-\pi}^{\pi}\log(1+G_0(\omega))\mathrm{d}\omega -\frac{1}{2\pi}\int_{-\pi}^{\pi}\log\left(1+\frac{|H(\omega)|^2}{N_0}\right)\mathrm{d}\omega \notag \\
&=&\frac{1}{2\pi}\int_{-\pi}^{\pi}\log\big(- M(\omega)\left(1+G_0(\omega)\right)\big)\mathrm{d}\omega \notag \\
&\leq&\log\left(- \frac{1}{2\pi}\int_{-\pi}^{\pi}M(\omega)\left(1+G_0(\omega)\right)\mathrm{d}\omega\right) \notag \\
&=& 0. \notag\eea}
\hspace{-1.2mm}Therefore, $I_{\mathrm{UBM}}\leq\mathcal{C}$ holds which completes the proof.

\bibliographystyle{IEEEtran}

\end{document}